\documentclass[journal]{IEEEtran}[12pt]
\ifCLASSINFOpdf
\usepackage[pdftex]{graphicx}
\graphicspath{{../pdf/}{../jpeg/}}
\DeclareGraphicsExtensions{.pdf,.jpeg,.png}
\else
\usepackage[dvips]{graphicx}
\graphicspath{{../eps/}}
\DeclareGraphicsExtensions{.eps}
\fi\usepackage{graphicx}

\usepackage{graphics}
\usepackage{epsfig}
\usepackage{epstopdf}
\usepackage{stfloats}
\usepackage[cmex10]{amsmath}
\usepackage{algorithmic}
\usepackage{array}
\usepackage{mdwmath}
\usepackage{mdwtab}
\usepackage{graphicx}
\usepackage{subfigure}
\usepackage{color}
\usepackage{amsfonts,amssymb}
\usepackage{hyperref} 
\usepackage{multirow}
\usepackage{diagbox}
\usepackage{caption}

\usepackage{amsfonts,amsthm,array} 

\newtheorem{theorem}{Theorem}

\newtheorem{corollary}{Corollary}

\newtheorem{remark}{Remark}

\begin{document}
	
\title{Outage Analysis of Aerial Semi-Grant-Free NOMA Systems}
	
\author{Hongjiang~Lei, 
		Chen~Zhu,
		Ki-Hong~Park, 
		Imran~Shafique~Ansari, \\
		Weijia Lei, 
		Hong~Tang,
		and~Kyeong~Jin~Kim 
		\thanks{This work was supported by the National Natural Science Foundation of China under Grant 61971080 and the Open Fund of the Shaanxi Key Laboratory of Information Communication Network and Security under Grant ICNS201807.}
		\thanks{Hongjiang Lei, Chen Zhu, Weijia Lei, and Hong Tang are with the School of Communication and Information Engineering, Chongqing University of Posts and Telecommunications, Chongqing 400065, China, also with Chongqing Key Lab of Mobile Communications Technology, Chongqing 400065, China, and H. Lei is also with Shaanxi Key Laboratory of Information Communication Network and Security, Xi'an University of Posts and Telecommunications, Xi'an, Shaanxi 710121, China (e-mail: leihj@cqupt.edu.cn).}
		\thanks{Ki-Hong~Park is with CEMSE Division, King Abdullah University of Science and Technology (KAUST), Thuwal 23955-6900, Saudi Arabia (e-mail: kihong.park@kaust.edu.sa).}
		\thanks{Imran~Shafique~Ansari is with James Watt School of Engineering, University of Glasgow, Glasgow G12 8QQ, United Kingdom (e-mail: imran.ansari@glasgow.ac.uk).}
		\thanks{Kyeong~Jin~Kim is with the Mitsubishi Electric Research Laboratories, Cambridge, MA 02139 USA (e-mail: kkim@merl.com).}		
}
	
\maketitle
\begin{abstract}

In this paper, we analyze the outage performance of Unmanned Aerial Vehicles (UAVs)-enabled downlink Non-Orthogonal Multiple Access (NOMA) communication systems with the Semi-Grant-Free (SGF) transmission scheme.
A UAV provides coverage services for a Grant-Based (GB) user and one Grant-Free (GF) user is allowed to utilize the same channel resource opportunistically.
The analytical expressions for the exact and asymptotic Outage Probability (OP) of the GF user are derived.
The results demonstrate that no-zero diversity order can be achieved only under stringent conditions on users' quality of service requirements.
Subsequently, an efficient Dynamic Power Allocation (DPA) scheme is proposed to relax such data rate constraints.
The analytical expressions for the exact and asymptotic OP of the GF user with the DPA scheme are derived.
Finally, Monte Carlo simulation results are presented to validate the correctness of the derived analytical expressions and demonstrate the effects of the UAV's location and altitude on the OP of the GF user.

\end{abstract}
	
\begin{IEEEkeywords}
Unmanned aerial vehicle, non-orthogonal multiple access, semi-grant-free, outage probability.
\end{IEEEkeywords}

\section{Introduction}
\label{sec:introduction}

\subsection{Background and related works}
In recent years, Unmanned Aerial Vehicles (UAVs) have been envisioned to play an essential role in space-air-ground integrated networks due to the flexibility of deployment, controllable mobility, and low costs \cite{WuQ2021JSAC}. 
Multiple access techniques are essential to integrate UAVs into 5G networks and beyond.
Non-Orthogonal Multiple Access (NOMA) is regarded as a profitable candidate for 5G networks because of its merit in providing higher Spectral Efficiency (SE) and supporting massive connectivity \cite{LiuY2019WC}.
Multiple users are served simultaneously in non-orthogonal channel resources by isolating the users in the power domain \cite{DingZ2017Mag}.
Using NOMA technology, UAVs can provide services for multiple users over the same resource block.
Quite a few current investigations have considered using NOMA to improve the performance of UAV-enabled communication systems.
Hou \emph{et al.} focused on addressing the spatial distribution problem of NOMA-enhanced UAV network by utilizing stochastic geometry tools \cite{HouT20191TCOM}-\cite{HouT2020IoT}.
In \cite{HouT20191TCOM}, a new 3D UAV framework for downlink wireless service to randomly roaming NOMA users was proposed.
Analytical expressions for the Outage Probability (OP) and Ergodic Capacity (EC) of Multiple-Input-Multiple-Output (MIMO)-NOMA-enhanced UAV networks were derived.
In \cite{HouT2020TVT}, the NOMA-enhanced UAV-to-everything networks were investigated. The closed-form expressions for the OP and EC of the paired NOMA receivers were derived.
In \cite{HouT20192TCOM}, the UAV-centric strategy for offloading actions and the user-centric strategy for providing emergency network were considered, and the analytical expressions of the Coverage Probability (CP) for both scenarios with imperfect Successive Interference Cancelation (SIC) were derived.

NOMA-enhanced terrestrial Internet of Things (IoT) networks and NOMA-enhanced aerial IoT networks were investigated in \cite{HouT2020IoT}, and new channel statistics were derived for both terrestrial and aerial users. 
Then, the analytical expressions for the exact and the asymptotic coverage probability were derived. 
{To maximize the sum capacity of NOMA-enabled backscatter communication systems, the transmit power of IoT users at BS and the reflection coefficient of backscatter tag were jointly optimized in \cite{KhanWU2021ITL}. The convex optimization problems were solved by Karush-Kuhn-Tucker conditions.}

Performance optimization of NOMA-aided UAV systems has been well-researched in many works.
In \cite{ZhaoN2019TCOM}, a UAV-assisted NOMA system was studied where the UAV assisted the Base Station (BS) in providing services to the ground users.
The sum rate was maximized by optimizing the UAV trajectory and the NOMA precoding. 
In \cite{ZhaiD2021TCOM}, an aerial Decode-and-Forward (DF) cooperative NOMA network was jointly optimized concerning UAV height, channel allocation, and power allocation to maximize the data rate.
Ref. \cite{ZhangH2020SAC} considered the Energy Efficiency (EE) of the UAV's communication with imperfect channel state information where the EE was maximized by designing user scheduling and power allocation.
Authors in \cite{WangW2021TCOM} proposed a time-efficient data collection scheme where multiple ground devices upload their data to the UAV via uplink NOMA. The duration of each time slot was minimized by jointly optimizing the trajectory, device scheduling, and transmit power.
In \cite {LiuM2020IoT}, the location of the UAV and power allocation were jointly optimized to enhance the EE and SE of NOMA-aided UAV systems.

Conventional wireless communication systems operate on Grant-Based (GB) protocols \cite{ShirvanimoghaddamM2017Mag}. 
Each device first transmits a scheduling request to BS and then sends a grant back for resource allocation.
The lengthy handshaking process of requesting a grant will be prohibitively costly for the signaling overhead and unacceptable for the resulting latency in massive IoT data transmissions \cite{ParkD2005TCOM}.
Ding \emph{et al.} proposed NOMA-assisted Semi-Grant-Free (SGF) transmission schemes and two contention control mechanisms have been presented to ensure that the number of users admitted to the same channel was carefully controlled in \cite{DingZ2021TCOM}.
And then, a new SGF transmission scheme combining the flexibility in choosing the decoding order in NOMA was proposed.
The closed-form expressions for the exact and asymptotic OP were derived. The results show that the NOMA-assisted SGF transmission scheme effectively addresses the problem of the aforementioned request-grant process and the spectrum reserved.
Based on \cite{DingZ2021TCOM}, a new Power Control (PC) strategy was proposed to guarantee no OP floor entirely by adjusting the transmit power of Grant-Free (GF) user to control the decoding order of SIC at the base station in \cite{SunY2021TVT}.
In \cite{YangZ2020WCL}, an adaptive power allocation strategy was proposed founded on the relationship between the target date rate of GB users and channel conditions of both GB and GF users.
Authors in \cite{LuH2022TWC} considered a NOMA system with multiple randomly distributed GF users. 
The analytical expressions for the OP with fixed transmit power and dynamic power control strategy were derived, and the small-scale fading, path loss, and random user locations were considered. 
Furthermore, the outage performance of GF users was analyzed under the best user scheduling scheme and Cumulative Distribution Function (CDF)-based scheduling scheme. 
In \cite{ZhangC2022TWC}, an uplink SGF NOMA system with multiple randomly deployed GF users was studied by utilizing stochastic geometry techniques. 
A dynamic protocol was proposed to interpret part of the GF users that are paired to NOMA groups. 
The outage performance and diversity gains under dynamic and open-loop protocols were investigated, and numerical results show that dynamic protocol effectively improves the outage performance.
Moreover, the analytical expressions for the exact and approximated EC were derived in \cite{ZhangC2021WCL}. 
{The secrecy performance of NOMA-aided SGF was investigated in \cite{LeiH2022TWC} wherein the analytical expressions for the Secrecy Outage Probability (SOP) for the scenarios with a single GF user and multiple GF users were derived, respectively.}
In \cite{ChenJ2022JSAC}, an Intelligent Reconfigurable Surface (IRS)-assisted SGF NOMA system was investigated, in which the IRS enhanced the channel gains for GB and GF users. The sum rates of GF users were maximized by jointly optimizing the sub-carrier assignment, the power allocation of GF users, and the IRS amplitude and phase shift.

{Table \ref{table1} outlines related works on performance analysis of NOMA-aided SGF systems.}
\begin{table*}[!ht]
\centering
\caption{\emph{{Recent literature related to performance analysis of NOMA-aided SGF systems}.}}
{\footnotesize
	{
		\begin{center}
			\begin{tabular}{| c | c | c | c | c |}
				\hline
				\textbf{Reference} &\textbf{downlink/uplink NOMA} &\textbf{UAV} &\textbf{Method} &\textbf{Performance Metrics}\\
				\hline
				\cite{HouT20191TCOM}               	 & downlink & \checkmark       & Performance analysis 		 & OP, EC \\
				\hline
				\cite{HouT2020TVT}               	 		& downlink & \checkmark       & Performance analysis 	    & OP, EC \\
				\hline
				\cite{HouT20192TCOM}               	 & downlink & \checkmark       & Performance analysis 		 & CP \\
				\hline
				\cite{HouT2020IoT}                	 		& uplink & \checkmark       & Performance analysis 	     & CP \\
				\hline
				\cite{KhanWU2021ITL}                	 & downlink &     				  & Optimization 					& Sum rate   \\
				\hline
				\cite{ZhaoN2019TCOM}                 & downlink & \checkmark     & Optimization    				  & Sum rate   \\
				\hline
				\cite{ZhaiD2021TCOM}                   & downlink & \checkmark    & Optimization 					   & Sum rate   \\
				\hline
				\cite{ZhangH2020SAC}                	& downlink & \checkmark    & Optimization 					& EE   \\
				\hline			
				\cite{WangW2021TCOM}                & uplink & \checkmark    & Optimization 					& Flight time   \\
				\hline
				\cite{LiuM2020IoT}                		     & uplink & \checkmark    & Optimization 					& SE, EE   \\
				\hline
				\cite{DingZ2021TCOM}                    & uplink, SGF &     & Performance analysis  				& OP \\
				\hline
				\cite{SunY2021TVT}                 	 	    & uplink, SGF &     & Performance analysis  				& OP     \\				
				\hline
				\cite{YangZ2020WCL}               		  & uplink, SGF &     & Performance analysis  				 & OP     \\
				\hline
				\cite{LuH2022TWC}                  			& uplink, SGF &     & Performance analysis  				& OP     \\
				\hline
				\cite{ZhangC2022TWC}                 	 & uplink, SGF &     & Performance analysis      			& OP  \\
				\hline
				\cite{ZhangC2021WCL}                	 & uplink, SGF &     & Performance analysis  				& EC\\
				\hline
				\cite{LeiH2022TWC}                	 		& uplink, SGF &     & Performance analysis  				& SOP\\
				\hline
				\cite{ChenJ2022JSAC}                	   & uplink, SGF &     & Deep reinforcement learning 	 & Average data rate    \\
				\hline			
			\end{tabular}
	\end{center}}
	\label{table1}	}
	\end{table*}
	
	\subsection{Motivation and contributions}
	
	In light of the above-discussed works, the performance of uplink NOMA systems with the SGF scheme has been studied in varying scenarios. However, there still needs to be more research contributions on investigating the performance of UAV communication systems with the SGF scheme, which motivates this work.
	In this work, we dedicate ourselves to developing new SGF schemes utilized in UAV-enabled downlink NOMA communication systems. The impacts of the UAV's location and altitude on outage performance are studied. 
	The main contributions of this paper are outlined as follows:
	\begin{enumerate}

\item We investigate the outage performance of the GF user and the analytical expression for the exact OP is derived. We also analyze the asymptotic OP and the achievable diversity orders in the higher-signal-to-noise ratio (SNR) region to obtain more insights. The results demonstrate that no-zero diversity order can be achieved only under stringent conditions on users' quality of service requirements.

\item An efficient Dynamic Power Allocation (DPA) scheme is proposed to guarantee no floor without any conditions on users' target rates. The analytical expressions for the exact and asymptotic OP of the GF user with the DPA scheme are derived. 
The results demonstrate no OP floor for all the users' target rates when the DPA scheme is utilized.

\item Monte Carlo simulation results are presented to validate the correctness of the derived analytical expressions and demonstrate the effects of the UAV's location and altitude on the OP of the GF user.

\item Relative to \cite{DingZ2021TCOM} - \cite{LuH2022TWC} wherein the performance of the uplink NOMA-aided SGF systems was investigated, the transmit power and the Channel State Information (CSI) of GB users must be known at the GF user to realize the power control.
This work studied the performance of the downlink NOMA-aided SGF systems, where the power allocation is utilized at the base station and the GF user need not know the CSI of the GB user.

\end{enumerate}

\subsection{Organization}
The rest of this paper is organized as follows. Section \ref{sec:SystemModel} describes the considered system model and the SGF scheme utilizing in downlink NOMA systems. 
The OP of the GF user with Fixed Power Allocation (FPA) and DPA are analyzed in Sections \ref{sec:FPA} and \ref{sec:DPA}, respectively. 
Section \ref{sec:RESULTS} presents the simulation results to demonstrate the analysis and the paper is concluded in Section \ref{sec:Conclusion}.

\section{System model}
\label{sec:SystemModel}

\subsection{System model}

\begin{figure}[t]
\centering
\includegraphics[width = 2.53in]{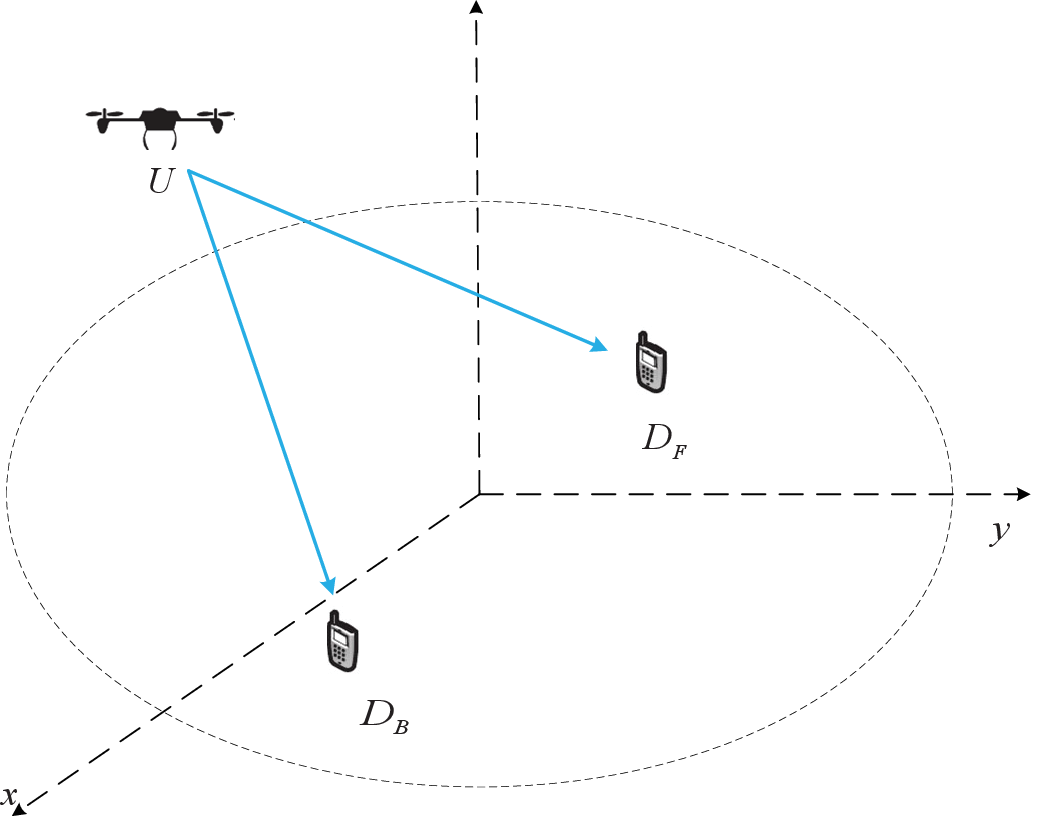}
\caption{A downlink UAV-based NOMA communication system consisting of one UAV ($U$) and two legitimate users (${D_B}$ and ${D_F}$).
}
\label{fig1model}
\end{figure}

As shown in Fig. \ref{fig1model}, we consider a UAV-enabled downlink NOMA system that consists of an aerial base station $\left( U \right)$, a GB user $\left( {{D_B}} \right)$, and a GF user 
$\left( {{D_F}} \right)$\footnote{ Although only two users are considered in this work, our results can be easily extended to NOMA systems with more than two users by utilizing the hybrid multiple access scheme proposed in \cite{YueX2018TCOMUN}, \cite{YueX2020TWC}, \cite{VaeziM2019WC}. }. 
Similar to \cite{DingZ2021TCOM}, $D_B$ is allocated to one dedicated resource block and $D_F$ will gain admission to the resource block opportunistically.
Without loss of generality, a 3D cartesian coordinate system with the origin at $O$ is utilized. 
The coordinates of $D_B$ and $D_F$ are denoted as $\left( {{x_B},{y_B},0} \right)$ and $\left( {{x_F},{y_F},0} \right)$, respectively. 
The coordinate of $U$ is denoted as $\left( {{x_U},{y_U}, {z_U}} \right)$ and the distance between $U$ and $D_X$ $\left( {X \in \left\{ {B,F} \right\}} \right)$ is expressed as
\begin{equation}
{d_X}= \sqrt {{{\left( {{x_X} - {x_U}} \right)}^2} + {{\left( {{y_X} - {y_U}} \right)}^2} + {{z_U^2}}}
\label{X_dn}
\end{equation}
Due to possible obstacles between the Air-to-Ground (A2G) links, Line-of-Sight (LoS) and Non-Line-of-Sight (NLoS) connections are probabilistically considered in this work. As a result, the average path loss between $U$ and $D_X$ is expressed as \cite{Zhou2018TVT}
\begin{equation}
{{\bar g}_{X}} = \left( {P_{X}^{\rm{L}}{\eta _{\rm{L}}} + P_{X}^{{\rm{nL}}}{\eta _{{\rm{nL}}}}} \right)d_{X}^{{\alpha _{X}}}
\label{average_path_loss0}
\end{equation}
where ${\eta _{{\rm{L}}}}$ and ${\eta _{{\rm{nL}}}}$ signify the attenuation factor to the LoS and NLoS links, respectively,
$P_X^{\rm{L}}$ signifies the probability of LoS connection,
$P_X^{\rm{L}} = {\left( {1 + {a_0}{e^{ - {b_0}\left( {\frac{{180}}{\pi }{\theta _X} - {a_0}} \right)}}} \right)^{ - 1}}$,
${{a_0}}$ and ${{b_0}}$ are environmental parameters as listed in Table I and II of 
\cite{Hourani2014WCL},
${\theta _{X}} = \arcsin \left( {\frac{H}{{{d_{X}}}}} \right)$ denotes elevation angle (in radians),
$P_{X}^{{\rm{nL}}}= 1 - P_{X}^{\rm{L}}$, and
${\alpha _X}$ is the path loss exponent.	
The relationship between ${\alpha _X}$ and ${\theta _X}$ is expressed as \cite{Azari2018TCOM}
\begin{equation}
{\alpha _X}\left( {{\theta _X}} \right) = {b_1}P_X^{\rm{L}} + {b_2}
\label{alphaX}
\end{equation}
where
${b_1} \approx {\alpha _{\frac{\pi }{2}}} - {\alpha _0}$
and
${b_2} \approx {\alpha _0}$.
In this work, we set ${\alpha _{\frac{\pi }{2}}} = 2$ and ${\alpha _0} = 4$.

It is assumed that the fading coefficients, $h_X$, in the A2G links experience independent and identically (i.i.d) Nakagami-$m$ fading. 
Denoting ${G_X} = \frac{{{h_{X}}}}{{{{\bar g}_{{X}}}}}$, the Probability Density Function (PDF) and the CDF of ${G_{X}}$ are expressed by
\begin{equation}
{F_{{G_X}}}\left( x \right) = 1 - {e^{ - {\lambda _X}x}}\sum\limits_{i = 0}^{{m} - 1} {\frac{{{{\left( {{\lambda _X}x} \right)}^i}}}{{i!}}}
\label{G_X_cdf}
\end{equation}
\begin{equation}
{f_{{G_X}}}\left( x \right) = \frac{{\lambda _X^m}}{{\Gamma \left( m \right)}}{x^{m - 1}}{e^{ - {\lambda _X}x}}
\label{G_X_pdf}
\end{equation}
where $m$ denotes the fading parameter that is integer,
${\lambda _X} = {m}{\bar g_X}$,
and $\Gamma \left(  \cdot  \right)$ is the Gamma function as defined by \cite[(8.310.1)]{GradshteynBook}.

\subsection{SGF schemes}
\label{sec:SGF Schemes}

To ensure the $D_B$'s Quality-of-Service (QoS) in the worse scenario, the $D_B$'s rate is constrained as \cite{DingZ2021TCOM}
\begin{equation}
{\log _2}\left( {1 + \frac{{\rho {\omega }{G_B}}}{{1 + \rho {\bar \omega }{G_B}}}} \right)  \ge  {R_{{\rm{th}}}^B}
\label{qosDB}
\end{equation}
where
$\rho  = \frac{P}{{{\sigma ^2}}}$,
$P$ is the transmit power of $U$,
${{\sigma ^2}}$ is the variance of complex zero mean Additive White Gaussian Noise (AWGN),
$0 \le \omega  \le 1$ denotes the power allocation coefficient for $D_B$,
$\bar \omega = 1 - \omega$,
and
$R_{{\rm{th}}}^B$ denotes the Reliability Rate Threshold (RRT) of $D_B$.
Then, we have
\begin{equation}
{\omega }  \ge  \frac{{\left( {\rho {G_B} + 1} \right)\left( {{\Theta _B} - 1} \right)}}{{\rho {G_B}{\Theta _B}}}
\label{pacDB}
\end{equation}
\begin{equation}
\rho  \ge \frac{{{\Theta _B} - 1}}{{{G_B}\left( {{\Theta _B}{\omega } - \left( {{\Theta _B} - 1} \right)} \right)}}
\label{rhoDB}
\end{equation}
\begin{equation}
{G_B} \ge \frac{{{\Theta _B} - 1}}{{\rho \left( {1 - {\Theta _B}{\bar \omega }} \right)}}
\label{GBDB}
\end{equation}
where
${\Theta _B} = {2^{{R_{{\rm{th}}}^B}}}$.
Eqs. (\ref{pacDB}), (\ref{rhoDB}), and (\ref{GBDB}) denote the conditions in which the power allocation coefficient, the transmission SNR, and the channel coefficient of $D_B$ must meet for the other given parameters.

\begin{remark}
When $\omega  = 1$, one can easily find that decoding its own signal would fail at $D_B$ if ${\log _2}\left( {1 + \rho {G_B}} \right) < {R_{{\rm{th}}}^B}$, which means to ensure that $D_B$ decodes its own signal successfully, there is a constraint
\begin{equation}
	{G_B} > {\varepsilon _1}
	\label{eq10}
\end{equation}
where
${\varepsilon _1} = \frac{{{\Theta _B} - 1}}{\rho }$.
\end{remark}

It must be noted that the SGF scheme only guarantees that admitting the GF user is transparent to the GB user whose QoS experience is the same when it occupies the channel alone \cite{LeiH2022TWC}. In other words, the SGF scheme does not always guarantee no outage for $D_B$.
When the fading over the link between $U$ and $D_B$ is too strong, the constraint in (\ref{eq10}) can not be satisfied thereby the SGF scheme can not be utilized because there are no signals to $D_F$ to avoid any performance degradation for $D_B$.

Considering the decoding order at $ D_F$ for given $\omega $ and $\rho$, 
according to the principle of SIC in NOMA, if the $D_B$'s signal can be successfully decoded in the first stage of SIC and deleted from the superimposed signal received at $D_F$, the interference can be eliminated and the maximum achievable rate can be obtained.
Because ${\log _2}\left( {1 + \frac{{\rho {\omega }x}}{{1 + \rho {\bar \omega }x}}} \right) = {\log _2}\left( {1 + \frac{{\rho {\omega }}}{{\frac{1}{x} + \rho {\bar \omega }}}} \right)$ is an increasing function of $x$ and ${\log _2}\left( {1 + \frac{{\rho {\omega }{G_F}}}{{1 + \rho {\bar \omega }{G_F}}}} \right) > {R_{{\rm{th}}}^B}$ holds when ${G_F} > {G_B}$.
Then, the condition of $D_F$ decoding $D_B$'s signal at the first stage of SIC is ${G_F} > {G_B}$.
According to this, there are two different decoding orders at $D_F$, which is stated as follows:
\begin{itemize}
\item Case 1: $U$ is located in the region where the channel condition of $D_F$ is stronger than that of $D_B$, namely ${G_F} > G_B$. In this scenario, $D_B$'s signal can be decoded at the first stage or the second stage of SIC. Accordingly, $D_F$ will achieve a data rate of  ${\log_2} \left( {1 + {\bar \omega }\rho {G_F}} \right)$ or $\log \left( {1 + \frac{{{\bar \omega }\rho {G_F}}}{{1 + \rho {\omega }{G_F}}}} \right)$. To maximize its data rate, $D_F$ will decode $D_B$'s signal at the first stage of SIC due to $\log \left( {1 + \frac{{{\bar \omega }\rho {G_F}}}{{1 + \rho {\omega }{G_F}}}} \right) < {\log_2} \left( {1 + {\bar \omega }\rho {G_F}} \right)$. Then, the achievable rate of $D_F$ in this scenario is expressed as
\begin{equation}
	R_F^1 = \log _2 \left( {1 + {\bar \omega }\rho {G_F}} \right)
	\label{rateF1}
\end{equation}
\item Case 2: $U$ is located in the region where the channel condition of $D_F$ is weaker than that of $D_B$, namely ${G_F} < G_B$. In this scenario, $D_F$ cannot decode $D_B$'s signal at the first stage of SIC since ${G_F} < G_B $ leads to ${\log _2}\left( {1 + \frac{{\rho {\omega }{G_F}}}{{1 + \rho {\bar \omega }{G_F}}}} \right) < {R_{{\rm{th}}}^B}$. Therefore, $D_F$ must decode its own signal firstly, which achieves the data rate as
\begin{equation}
	R_F^2 = {\log _2}\left( {1 + \frac{{{\bar \omega }\rho {G_F}}}{{1 + \rho {\omega }{G_F}}}} \right)
	\label{rateF2}
\end{equation}
\end{itemize}

Then, the achievable rate of $D_F$ is expressed as
\begin{equation}
{R_F} = \left\{ {\begin{array}{*{20}{c}}
		{R_F^1,}&{{G_F} > {G_B}} \\
		{R_F^2,}&{{G_F} < {G_B}}
\end{array}} \right.
\label{rateF}
\end{equation}

\section{Outage performance analysis with fixed power allocation}
\label{sec:FPA}

Note that when the SGF scheme is adopted, the QoS requirements of the GB user will be met first. 
The GF user can utilize GB's channel only when the GB user can achieve its target rate. 
If the outage happens to the GB user, it means that the channel resource is insufficient to meet the QoS requirements of the GB user; therefore, the GF user is not allowed to access the channel, and the GB user occupies the channel alone. 
Thus, $D_B$ can consistently achieve the same outage performance as in Orthogonal Multiple Access (OMA) and so we only investigate $D_F$'s outage performance.

In this section, it is assumed that $D_B$ is allocated to a fixed power $\omega = \min \left\{ {\frac{{\left( {\rho {G_B} + 1} \right)\left( {{\Theta _B} - 1} \right)}}{{\rho {G_B}{\Theta _B}}},1} \right\}$ that gives priority to meet its QoS requirement. To distinguish the scheme proposed in Section \ref{sec:DPA}, this scheme is termed as FPA scheme.
The OP of $D_F$ with the SGF scheme is expressed by
\begin{equation}
{P_{{\rm{out}}}} = \underbrace {\Pr \left\{ {{G_B} < {\varepsilon _1}} \right\}}_{ \triangleq {T_0}} + \underbrace {\Pr \left\{ {{G_B} > {\varepsilon _1},R_F < {R_{{\rm{th}}}^F}} \right\}}_{ \triangleq {T_1}}
\label{out1}
\end{equation}
where
${R_{{\rm{th}}}^F}$ signifies the RRT of $D_F$,
$T_0$ denotes $D_B$ is in outage when it occupies the channel alone, which means $D_F$  is not allowed to utilize the channel,
and $T_1$ denotes $D_F$ is in outage when it is allowed to utilize $D_B$'s channel.

\begin{remark}
When $\rho \to \infty$, there is ${\varepsilon _1} \to 0$, then $T_0 \to 0$. Thus, ${{T_0}}$ is the main part of ${P_{out}}$ in the lower-$\rho$ region while ${{T_1}}$ is the main part of ${P_{out}}$ in larger-$\rho$ region.
\end{remark}

The following theorem provides the analytical expression for the OP of $D_F$.
\begin{theorem}
\label{theorem1}
The analytical expression for the OP of $D_F$ is expressed as	
\begin{equation}
	P_{{\rm{out}}}^{{\rm{FPA}}} = \left\{ {\begin{array}{*{20}{c}}
			{{T_0} + {T_{11}} + {T_{12a}},}&{{\Theta _{th}} < \frac{{{\Theta _B}}}{{{\Theta _B} - 1}}}\\
			{{T_0} + {T_{11}} + {T_{12b}},}&{{\Theta _{th}} > \frac{{{\Theta _B}}}{{{\Theta _B} - 1}}}
	\end{array}} \right.
	\label{theorem1eq}
\end{equation}
where
${\Theta _{th}} = {2^{R_{{\rm{th}}}^F}}$,
${T_0} = {F_{{G_B}}}\left( {{\varepsilon _1}} \right)$,
${T_{11}} = {\chi _1} - {\chi _2}$,
${T_{12a}} = {\chi _3} + {\chi _4}$,
${T_{12b}} = {{\bar F}_{{G_B}}}\left( {{\varepsilon _1}} \right) - {A_1}\sum\limits_{i = 0}^{m - 1} {\frac{{\lambda _F^i\Gamma \left( {i + m,{A_2}{\varepsilon _1}} \right)}}{{i!{A_2}^{i + m}}}}$,
${\chi _1} = {F_{{G_B}}}\left( {{\varepsilon _0}} \right) - {F_{{G_B}}}\left( {{\varepsilon _1}} \right) - {A_1}{\Phi _1}$,
${\chi _2} = {F_{{G_B}}}\left( {{\varepsilon _0}} \right) - {F_{{G_B}}}\left( {{\varepsilon _1}} \right) - {A_1}\sum\limits_{i = 0}^{m - 1} {\frac{{\lambda _F^i{\Phi _2}}}{{i!}}}$,
${\chi _3} = {F_{{G_B}}}\left( {{\varepsilon _5}} \right) - {F_{{G_B}}}\left( {{\varepsilon _1}} \right) - {A_1}\sum\limits_{i = 0}^{m - 1} {\frac{{\lambda _F^i{\Phi _3}}}{{i!}}}$,
${\chi _4} =  {{\bar F}_{{G_B}}}\left( {{\varepsilon _5}} \right) - {A_1}{\Phi _4}$,
${\Phi _1} = {g_1}\left( {{\varepsilon _1},{\varepsilon _2},{\varepsilon _1},{\varepsilon _0}} \right)$,
${\Phi _2} = \frac{{\Upsilon \left( {i + m,{A_2}{\varepsilon _0}} \right)}}{{A_2^{i + m}}} - \frac{{\Upsilon \left( {i + m,{A_2}{\varepsilon _1}} \right)}}{{A_2^{i + m}}}$,
${\Phi _3} = \frac{1}{{A_2^{i + m}}}\left( {\Upsilon \left( {i + m,{A_2}{\varepsilon _5}} \right) - \Upsilon \left( {i + m,{A_2}{\varepsilon _1}} \right)} \right)$,
$\Upsilon \left( { \cdot , \cdot } \right)$ is lower incomplete Gamma function as defined by  \cite[(8.350.2)]{GradshteynBook},
${\Phi _4} = {g_2}\left( {{\varepsilon _3},{\varepsilon _4},{\varepsilon _5}} \right)$,
${g_1}\left( {a,b,s,t} \right) = \frac{\pi }{N}\sum\limits_{i = 0}^{m - 1} {\frac{{{{\left( {{\lambda _F}b} \right)}^i}}}{{i!}}\sum\limits_{n = 1}^N {\frac{{{{\left( {{\mu _n}\left( {s,t} \right)} \right)}^{m + i - 1}}}}{{{{\left( {{\mu _n}\left( {s,t} \right) - a} \right)}^i}}}\sqrt {\left( {{\mu _n}\left( {s,t} \right) - s} \right)\left( {t - {\mu _n}\left( {s,t} \right)} \right)} } }$\\${e^{ - {\lambda _B}{\mu _n}\left( {s,t} \right) - {\lambda _F}b\frac{{{\mu _n}\left( {s,t} \right)}}{{{\mu _n}\left( {s,t} \right) - a}}}}$,
${\varepsilon _2} = \frac{{{\Theta _B}\left( {{\Theta _{th}} - 1} \right)}}{\rho }$,
${\varepsilon _0} = {\varepsilon _1} + {\varepsilon _2}$,
${g_2}\left( {a,b,c} \right) = \sum\limits_{i = 0}^{m - 1} {\frac{{{{\left( {{\lambda _F}b} \right)}^i}}}{{i!}}} \sum\limits_{n = 1}^N {\left( {\frac{{{\omega _n}\iota _n^{m + i - 1}}}{{{{\left( {{\iota _n} - a} \right)}^i}}}{e^{{\iota _n} - \left( {{\lambda _B}{\iota _n} + \frac{{{\lambda _F}b{\iota _n}}}{{{\iota _n} - a}}} \right)}}} \right.} \\
\left. { - \frac{{\pi {e^{ - \left( {{\lambda _B}{\mu _n}\left( {0,c} \right) + \frac{{{\lambda _F}b{\mu _n}\left( {0,c} \right)}}{{{\mu _n}\left( {0,c} \right) - a}}} \right)}}{{\left( {{\mu _n}\left( {0,c} \right)} \right)}^{m + i - 1}}\sqrt {{\mu _n}\left( {0,c} \right)\left( {c - {\mu _n}\left( {0,c} \right)} \right)} }}{{N{{\left( {{\mu _n}\left( {0,c} \right) - a} \right)}^i}}}} \right)$,
${\varepsilon _3} = \frac{{{\Theta _B} - 1}}{\rho }\frac{{{\Theta _{th}}}}{{{\Theta _{th}} - \left( {{\Theta _{th}} - 1} \right){\Theta _B}}}$,
${\varepsilon _4} = \frac{{{\Theta _B}\left( {{\Theta _{th}} - 1} \right)}}{{\left( {{\Theta _B} - {\Theta _{th}}\left( {{\Theta _B} - 1} \right)} \right)\rho }}$,
${\varepsilon _5} = \frac{{2{\Theta _B}{\Theta _{th}} - {\Theta _{th}} - {\Theta _B}}}{{\rho \left( {{\Theta _{th}} + {\Theta _B} - {\Theta _{th}}{\Theta _B}} \right)}}$,
${{\bar F}_{{G_X}}}\left( x \right) = 1 - {F_{{G_X}}}\left( x \right)$,
${A_1} = \frac{{\lambda _{{B}}^m}}{{\Gamma \left( m \right)}}$,
${A_2} = {\lambda _{{B}}} + {\lambda _{{F}}}$,
${\mu _n}\left( {s,t} \right) = \frac{{t + s}}{2} + \frac{{\left( {t - s} \right){\tau _n}}}{2}$,
${\tau _n} = \cos \frac{{\left( {2n - 1} \right)\pi }}{{2N}}$,
$N$ is the summation terms, which reflects accuracy vs. complexity,
$\iota _n$ is the $n$th zeros of Laguerre polynomials,
and
${w_n}$ is the Gaussian weight, which are given in Table (25.9) of \cite{Abramowitz1972Book}.

	\textbf{\textit{Proof}}:		Please refer to Appendix \ref{prooftheorem1}.
	\end{theorem}

	Eq. (\ref{theorem1eq}) provides an exact relationship between the OP of $D_F$ and all the system parameters. 
	Some interesting insights can be obtained. 
	Firstly, the requirement of $D_B$ directly determines whether the resources can be shared with $D_F$. 
	A higher RRT denotes a lower probability of resource sharing. 
	Secondly, the location of $S$ is one of the important factors affecting DF performance because that affects the decoding order on $D_F$. 
	Thus, the OP of $D_F$ can be minimized by jointly optimizing the location of $S$ and other systems parameters, which will be part of future works. 
	Last but not least, the relationship between the RRT for $D_B$ and $D_F$ users under the condition in which the resources can be shared also significantly impacts the OP of the $D_F$.
	The analytical expressions presented in Theorem \ref{theorem1} are complicated because many coupled factors affect the outage performance of the GF user, specifically, the decoding order, the RRT of $D_B$, the RRT of $D_F$, and the relationship between $D_B$'s and $D_F$'s channel condition. 
	
	The asymptotic OP is derived in the following corollary to obtain more insights.

	\newcounter{TempEqCnt0}
	\setcounter{TempEqCnt0}{\value{equation}}
	\setcounter{equation}{18} 
	\begin{figure*}[ht]
\begin{equation}
	\begin{aligned}
		P_{{\rm{out}}}^{{\rm{DPA}}} &= \underbrace {\Pr \left\{ {{G_B} < {\varepsilon _1}} \right\}}_{ \triangleq {T_0}} + \underbrace {\Pr \left\{ {{G_B} > {\varepsilon _1},{G_F} > {G_B},R_F^1 < {R_{{\rm{th}}}^F}} \right\}}_{ \triangleq {T_{11}}} \\
		&+ \underbrace {\Pr \left\{ {{G_B} > {\varepsilon _1},{G_F} < \frac{{{\Theta _B}{G_B}}}{{\rho {G_B} + 1}},R_F^2 < {R_{{\rm{th}}}^F}} \right\}}_{ \triangleq {T_2}} + \underbrace {\Pr \left\{ {{G_B} > {\varepsilon _1},\frac{{{\Theta _B}{G_B}}}{{\rho {G_B} + 1}} < {G_F} < {G_B},R_F^3 < {R_{{\rm{th}}}^F}} \right\}}_{ \triangleq {T_3}}
		\label{opDPA}
	\end{aligned}
\end{equation}
\hrulefill
\end{figure*}
\setcounter{equation}{\value{TempEqCnt0}}
\newcounter{TempEqCnt1}
\setcounter{TempEqCnt1}{\value{equation}}
\setcounter{equation}{20} 
\begin{figure*}[ht]
\begin{equation}
	P_{{\rm{out}}}^{{\rm{DPA}},\infty} = \left\{ {\begin{array}{*{20}{c}}
			{T_0^\infty  + T_{11}^\infty  + T_3^\infty ,}&{{\Theta _{th}} < \frac{{{\Theta _B}}}{{{\Theta _B} - 1}}} \\
			{T_0^\infty  + T_{11}^\infty  + T_{2b}^\infty  + T_3^\infty ,}&{{\Theta _{th}} > \frac{{{\Theta _B}}}{{{\Theta _B} - 1}}}
	\end{array}} \right.
	\label{corollary21}
\end{equation}
\hrulefill
\end{figure*}
\setcounter{equation}{\value{TempEqCnt1}}

\begin{corollary}
\label{corollary1}
When $\rho  \to \infty $, the asymptotic OP is expressed as
\begin{equation}
	\begin{aligned}
		P_{{\rm{out}}}^{{\rm{FPA,}}\infty } = \left\{ {\begin{array}{*{20}{c}}
				{T_0^\infty  + T_{11}^\infty  + T_{12a}^\infty ,}&{{\Theta _{th}} < \frac{{{\Theta _B}}}{{{\Theta _B} - 1}}}\\
				{T_0^\infty  + T_{11}^\infty  + T_{12b}^\infty ,}&{{\Theta _{th}} > \frac{{{\Theta _B}}}{{{\Theta _B} - 1}}}
		\end{array}} \right.
	\end{aligned}
	\label{approxi1}
\end{equation}
where $T_0^\infty  = \frac{{{{\left( {{\lambda _B}{\varepsilon _1}} \right)}^m}}}{{m!}}$,
$T_{11}^\infty  \approx \frac{{\lambda _B^m\left( {\varepsilon _0^m - {\varepsilon _1}^m} \right)}}{{m!}}\left( {\frac{{{{\left( {{\lambda _F}{\varepsilon _2}} \right)}^m}}}{{m!}} - 1} \right) + \frac{{\lambda _B^m}}{{\Gamma \left( m \right)}}\sum\limits_{i = 0}^{m - 1} {\frac{{\lambda _F^i\left( {\varepsilon _0^{i + m} - {\varepsilon _1}^{i + m}} \right)}}{{i!\left( {i + m} \right)}}}$,
$T_{12a}^\infty  = \chi _3^\infty  + \chi _4^\infty$,
$\chi _3^\infty  \approx \frac{{\lambda _B^m\left( {\varepsilon _5^m - \varepsilon _1^m} \right)}}{{m!}} - {A_1}\sum\limits_{i = 0}^{m - 1} {\frac{{\lambda _F^i\left( {\varepsilon _5^{i + m} - \varepsilon _1^{i + m}} \right)}}{{i!\left( {i + m} \right)}}}$,
$\chi _4^\infty  \approx \frac{{{{\left( {{\lambda _F}{\varepsilon _4}} \right)}^m}}}{{m!}}\left( {1 - \frac{{{{\left( {{\lambda _B}{\varepsilon _5}} \right)}^m}}}{{m!}}} \right)$,
and
$T_{12b}^\infty  \approx 1 - \frac{{{{\left( {{\lambda _B}{\varepsilon _1}} \right)}^m}}}{{m!}} - {A_1}\sum\limits_{i = 0}^{m - 1} {\frac{{\lambda _F^i}}{{i!{A_2}^{i + m}}}\left( {\Gamma\left( {i + m} \right) - \frac{{{{\left( {{A_2}{\varepsilon _1}} \right)}^{i + m}}}}{{i + m}}} \right)}$.

\end{corollary}

\textbf{\textit{Proof}}:	Please refer to Appendix \ref{proofCorollary1}.

Due to ${\varepsilon _1} \to 0$, ${\varepsilon _2} \to 0$,  ${\varepsilon _4} \to 0$, ${\varepsilon _5} \to 0$ when $\rho \to \infty$, we have $T_0^\infty  \to 0$, $T_{11}^\infty  \to 0$, $T_{12a}^\infty  \to 0$, $T_{12b}^\infty  \approx 1 - {A_1}\sum\limits_{i = 0}^{m - 1} {\frac{{\lambda _F^i\Gamma\left( {i + m} \right)}}{{i!{A_2}^{i + m}}}} $, which is a constant  independent of $\rho$.

\begin{remark}
{According to} the results presented in Corollary \ref{corollary1}, one can realize there is an OP floor when ${\Theta _{th}} > 1 + \frac{1}{{{\Theta _B} - 1}}$.
\end{remark}

{Utilizing} ${G_{\rm{d}}} =  - \mathop {\lim }\limits_{\rho  \to \infty } \frac{{\log P_{\rm{out}}^\infty \left( \rho  \right)}}{{\log \rho }}$, the diversity order with the SGF scheme is obtained as
\begin{equation}
G_{\rm{d}}^{{\rm{FPA}}} = \left\{ {\begin{array}{*{20}{c}}
		{m,}&{{\Theta _{th}} < \frac{{{\Theta _B}}}{{{\Theta _B} - 1}}}\\
		{0,}&{{\Theta _{th}} > \frac{{{\Theta _B}}}{{{\Theta _B} - 1}}}
\end{array}} \right.
\label{GDFPA}
\end{equation}

\section{Outage performance analysis with dynamic power allocation}	
\label{sec:DPA}

In this section, a DPA scheme is proposed to avoid the OP floor and the analytical expression for the OP with the DPA scheme is derived.

\subsection{Proposed dynamic power allocation scheme}

In the previous analysis, $D_B$ is always allocated to a fixed power $\omega $ that just gives priority to meet its QoS requirement, while the other power is allocated to $D_F$.
One can observe from Corollary \ref{corollary1} that there is OP floor when ${\Theta _{th}} > 1 + \frac{1}{{{\Theta _B} - 1}}$.
Recalling when $G_F > G_B$, $D_B$'s signal can be decoded firstly at $D_F$ and the achievable rate of $D_F$ is expressed as $R_F^1 = {\log _2}\left( {1 + \rho {\bar \omega }{G_F}} \right)$.
If $G_F < G_B$, $D_F$ has not enough capacity to decode $D_B$'s signal at the first stage of SIC, the achievable rate of $D_F$ is expressed as $R_F^2 = {\log _2}\left( {1 + \frac{{\rho {\bar \omega }{G_F}}}{{1 + \rho {\omega }{G_F}}}} \right)$.
It must be noted that in the scenarios wherein $G_F < G_B$, $U$ can increase $\omega $ to ${\omega _2}$ to make sure that $D_F$ can decode $D_B$'s signal at the first stage of SIC, then the achievable rate will be changed from $R_F^2 = {\log _2}\left( {1 + \frac{{\rho {\bar \omega }{G_F}}}{{1 + \rho {\omega }{G_F}}}} \right)$ to $R_F^3 ={\log _2}\left( {1 + \rho {{\bar \omega }_2}{G_F}} \right)$, where ${{\bar \omega }_2} = 1 - {\omega _2}$ and ${\omega _2} = 1 - \frac{{\rho {G_F} - \left( {{\Theta _B} - 1} \right)}}{{\rho {\Theta _B}{G_F}}}$.
The goal of the DPA scheme is to maximize $D_F$'s achievable rate.
Hence, the power allocation coefficient must be chosen according to the relationship between $R_F^2$ and $R_F^3$.
Specifically, if $R_F^2 > R_F^3$, the power allocation coefficient for $D_B$ is $\omega $ otherwise ${\omega _2}$.
Due to $R_F^2 < R_F^3  \Leftrightarrow {\frac{{{\Theta _B}{G_B}}}{{\rho {G_B} + 1}} < {G_F} < {G_B}}$, 
the achievable rate of $D_F$ with the DPA scheme is expressed as
\begin{equation}
{R_F} = \left\{ {\begin{array}{*{20}{c}}
		{R_F^1,}&{{G_F} > {G_B}} \\
		{R_F^2,}&{{G_F} < \frac{{{\Theta _B}{G_B}}}{{\rho {G_B} + 1}}} \\
		{R_F^3,}&{\frac{{{\Theta _B}{G_B}}}{{\rho {G_B} + 1}} < {G_F} < {G_B}}
		\label{rateDPA}
\end{array}} \right.
\end{equation}

\subsection{Outage performance analysis with dynamic power allocation scheme}
{On the ground of }(\ref{rateDPA}), the OP with the DPA scheme is expressed as (\ref{opDPA}), shown at the top of the this page,
\setcounter{equation}{19}
where $T_0$ and $T_{11}$ are given in Theorem \ref{theorem1}, $T_2$ denotes $D_F$'s signal that is decoded at the first stage when power allocation coefficient is ${\omega }$,
and $T_3$ denotes $D_F$'s signal that is decoded at the second stage when power allocation coefficient is ${\omega _2}$.

The following theorem provides the analytical expression for the OP of $D_F$  with the DPA scheme.
\begin{theorem}
The OP of $D_F$ with the DPA scheme, $P_{{\text{out}}}^{{\text{dp}}}$, is expressed as
\begin{equation}
	P_{{\rm{out}}}^{{\rm{DPA}}} = \left\{ {\begin{array}{*{20}{c}}
			{{T_0} + {T_{11}} + T_{2a}^a + {T_3},}&{{\Theta _B} > \frac{1}{{{\Theta _{th}} - 1}}} \\
			{{T_0} + {T_{11}} + T_{2a}^b + {T_3},}&{{\Theta _B} < \frac{1}{{{\Theta _{th}} - 1}}}
	\end{array}} \right.
	\label{theorem2eq}
\end{equation}
where
$T_{2a}^a = {{\bar F}_{{G_B}}}\left( {{\varepsilon _1}} \right) - {A_1}{\Phi _5}$,
$T_{2a}^b = {{\bar F}_{{G_B}}}\left( {{\varepsilon _1}} \right) - {A_1}{g_1}\left( { - \frac{1}{\rho },\frac{{{\Theta _B}}}{\rho },{\varepsilon _1},{\varepsilon _3}} \right) - {A_1}{g_1}\left( {{\varepsilon _3},{\varepsilon _4},{\varepsilon _3},{\varepsilon _6}} \right) - {A_1}{g_2}\left( { - \frac{1}{\rho },\frac{{{\Theta _B}}}{\rho },{\varepsilon _6}} \right)$,
${{\text{T}}_3} = {A_1}{\Phi _5} - {{\bar F}_{{G_F}}}\left( {{\varepsilon _0}} \right){{\bar F}_{{G_B}}}\left( {{\varepsilon _0}} \right) - {A_1}\sum\limits_{i = 0}^{m - 1} {\frac{{{{\left( {{\lambda _F}} \right)}^i}{\Phi _2}}}{{i!}}}$,
${\Phi _5} = {g_2}\left( { - \frac{1}{\rho },\frac{{{\Theta _B}}}{\rho },{\varepsilon _1}} \right)$,
and
${\Phi _6} = \frac{{\Upsilon \left( {i + m,{A_2}{\varepsilon _0}} \right) - \Upsilon \left( {i + m,{A_2}{\varepsilon _1}} \right)}}{{A_2^{i + m}}}$.

\textbf{\textit{Proof}}:	Please refer to Appendix \ref{prooftheorem2}.

\end{theorem}

Eq. (\ref{theorem2eq}) provides an exact relationship between the OP of $D_F$ with the DPA scheme and the system parameters. 
In addition to the insights mentioned in Theorem 1, an interesting phenomenon can be found. 
Intuitively, increasing the transmit power of $D_B$ tends to stronger inter-user interference on $D_F$, which deteriorates the performance of $D_F$.
However, in the DPA scheme, under some conditions, appropriately increasing the transmit power of $D_B$ ensures the signal of $D_B$ can be decoded successfully on $D_F$. The achievable rate of $D_F$ can be enhanced through the SIC technique; thus the performance is improved. 
Similar to \cite{SunY2021TVT} and \cite{LeiH2022TWC}, the collaboration among the GB and GF users not only improves resource utilization, but also enhances the GF users' performance while not affecting the GB's QoS.

To obtain more insights, we derive the analytical expressions for the asymptotic OP of $D_F$ with the DPA scheme.
\begin{corollary}
\label{corollary2}
When $\rho  \to \infty $, the asymptotic OP of $D_F$ is expressed as (\ref{corollary21}), shown at the top of this page,
\setcounter{equation}{21} 	
where
$T_{2b}^\infty  = \frac{1}{{m!}}{\left( {\frac{{{\lambda _F}{\Theta _B}}}{\rho }} \right)^m}\left( {1 - \frac{{{{\left( {{\varepsilon _1}{\lambda _B}} \right)}^m}}}{{m!}}} \right)$,
$T_3^\infty  \approx \frac{{\lambda _B^m\left( {\varepsilon _0^m - \varepsilon _1^m} \right)}}{{m!}} + \frac{{\lambda _F^m}}{{m!}}\left( {\varepsilon _0^m - \frac{{\Theta _B^m}}{{{\rho ^m}}}} \right) - {A_1}\sum\limits_{i = 0}^{m - 1} {\frac{{\lambda _F^i\Phi _2^\infty }}{{i!}}}$,
and
$\Phi _2^\infty  = \frac{{\varepsilon _0^{i + m} - \varepsilon _1^{i + m}}}{{i + m}}$.

\textbf{\textit{Proof}}:		Please refer to Appendix \ref{proofCorollary2}.

\end{corollary}

Because of $T_{0}^\infty  \propto {\rho ^{ - m}}$, $T_{11}^\infty  \propto {\rho ^{ - m}}$, $T_{2b}^\infty  \propto {\rho ^{ - m}}$, $T_{3}^\infty  \propto {\rho ^{ - m}}$, we obtain $P_{{\rm{out}}}^{{\rm{DPA}},\infty}  \propto {\rho ^{ - m}}$.

\begin{remark}
It can be realized from Corollary \ref{corollary2} that OP floors can be avoided when DPA scheme is adopted.
\end{remark}

\begin{figure}[t]
\centering
\subfigure[OP with FPA in suburban environments for varying $\rho$ and $y_U$.]{
	\label{fig02a}
	\includegraphics[width = 0.2285 \textwidth]{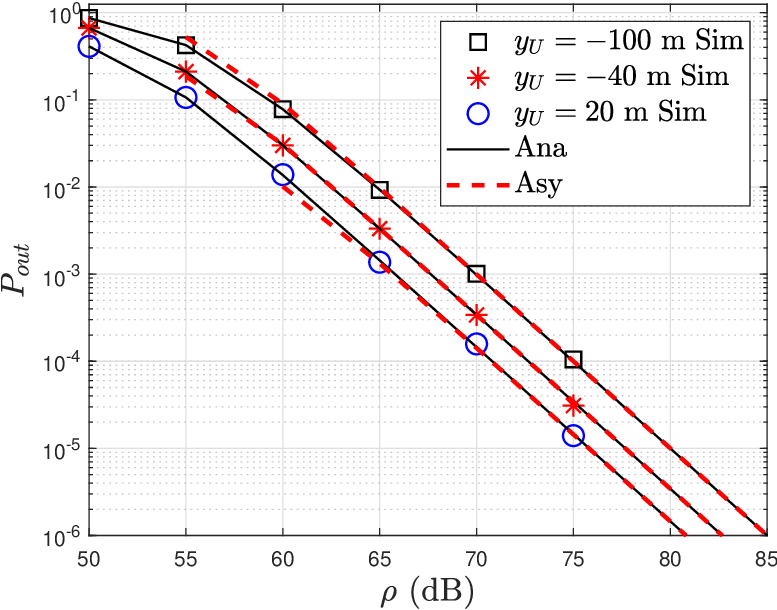}}
\subfigure[OP with DPA in suburban environments for varying $\rho$ and $y_U$.]{
	\label{fig02b}
	\includegraphics[width = 0.2285 \textwidth]{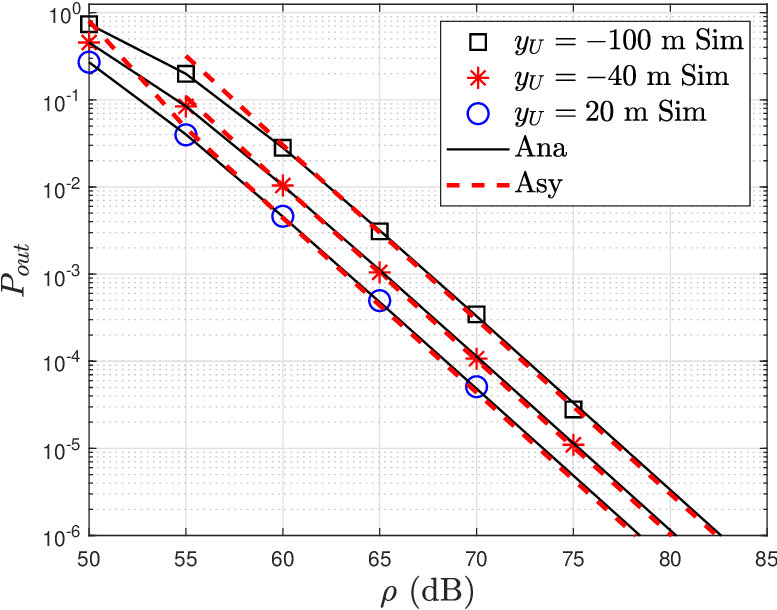}}
\subfigure[OP with FPA in urban environments for varying $\rho$ and $y_U$.]{
	\label{fig02c}
	\includegraphics[width = 0.2285 \textwidth]{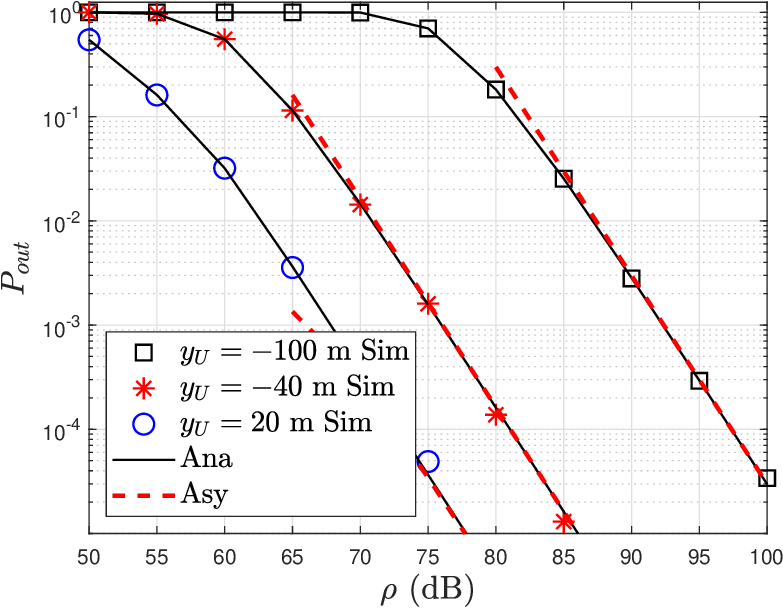}}
\subfigure[OP with DPA in urban environments for varying $\rho$ and $y_U$.]{
	\label{fig02d}
	\includegraphics[width = 0.2285 \textwidth]{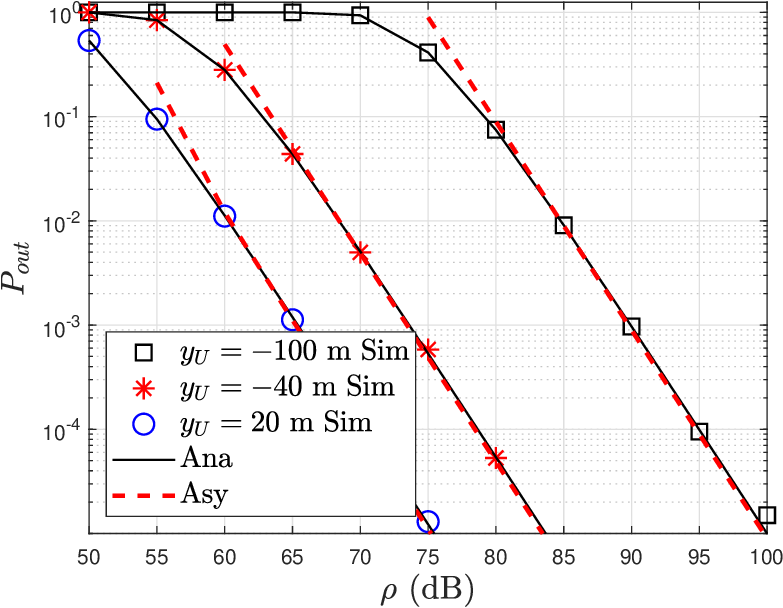}}
\caption{ {The effect of UAV's position and $\rho$ on the OP of $D_F$ with ${\Theta _{th}} < 1 + \frac{1}{{{\Theta _B} - 1}}$}.}
\label{fig02}
\end{figure}

Similar to (\ref{GDFPA}), the diversity order with DPA scheme is obtained as
\begin{equation}
G_{\rm{d}}^{{\rm{DPA}}} = m
\label{GDDPA}
\end{equation}

\section{Numerical results}
\label{sec:RESULTS}

In this section, Monte-Carlo simulations are presented to prove our analysis on the outage performance of the aerial SGF NOMA system by varying the parameters, such as transmit SNR and power allocation coefficient. The main parameters are set as $m = 2$, ${R_{{\rm{th}}}^B} = 0.2$ bps/Hz, ${R_{{\rm{th}}}^F} = 2$ bps/Hz, $\left( {{x_B},{y_B}} \right) = \left( {50, - 50} \right)$, $\left( {{x_F},{y_F}} \right) = \left( {50,50} \right)$,
and
$\left( {{x_U},{y_U},{z_U}} \right) = \left( {0,0,100} \right)$,
unless stated otherwise. In all the figures, `Sim', `Ana', and `Asy' denote the simulation, numerical, and asymptotic results, respectively.

Fig. \ref{fig02} presents the effect of UAV's position and $\rho$ on the OP of $D_F$ when ${\Theta _{th}} < 1 + \frac{1}{{{\Theta _B} - 1}}$.
It can be observed that OP is improved with the increase of transmission SNR. At the same time, by comparing Fig. \ref{fig02a} and Fig. \ref{fig02b}, it can also be observed that the curve of GF user's OP with DPA is always below the user OP curve with FPA, which demonstrates that the designed DPA scheme can effectively enhance the outage performance.
By comparing Fig. \ref{fig02a} and Fig. \ref{fig02c}, it can be observed that the outage performance of $D_F$ in suburban environments outperforms that in urban environments because there is a high LoS probability due to fewer shelters in suburban environments, which leads to a lower attenuation with minor average path loss. Thus, the degradation due to the variation of elevation angle is much severe in urban environment. 
The same conclusion can also be found by comparing Fig. \ref{fig02c} and Fig. \ref{fig02d}, and Fig. \ref{fig02b} and Fig. \ref{fig02d}.

\begin{figure}[t]
\centering
\subfigure[OP with FPA in suburban environments for varying $\rho$ and $y_U$.]{
\label{fig03a}
\includegraphics[width = 0.2285 \textwidth]{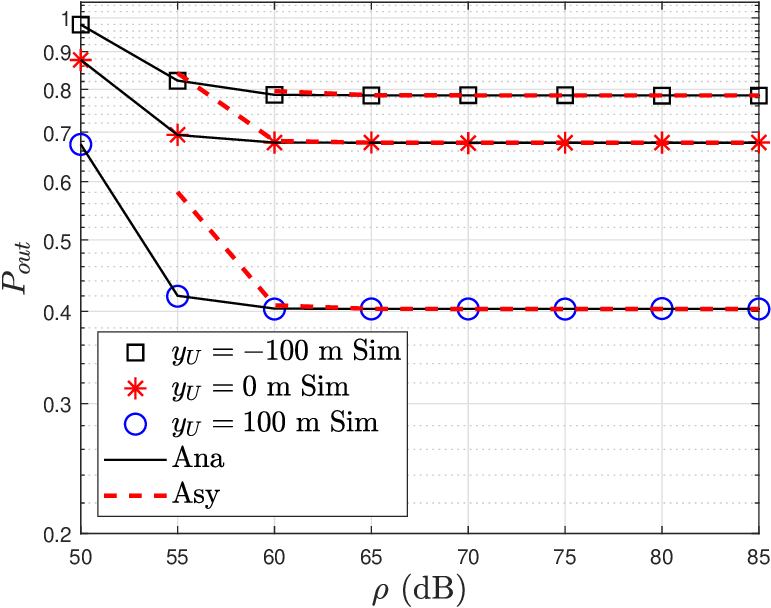}}
\subfigure[OP with DPA in suburban environments for varying $\rho$ and $y_U$.]{
\label{fig03b}
\includegraphics[width = 0.2285 \textwidth]{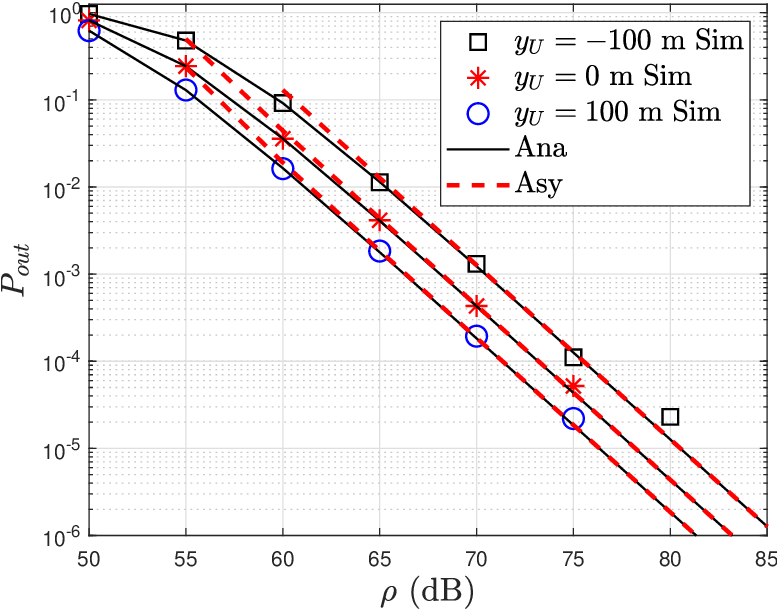}}
\subfigure[OP with FPA in urban environments for varying $\rho$ and $y_U$.]{
\label{fig03c}
\includegraphics[width = 0.2285 \textwidth]{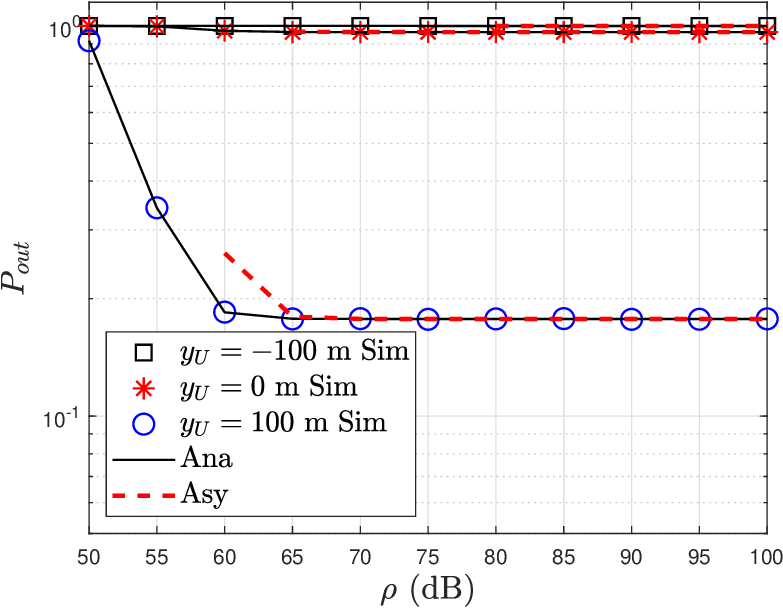}}
\subfigure[OP with DPA in urban environments for varying $\rho$ and $y_U$.]{
\label{fig03d}
\includegraphics[width = 0.2285 \textwidth]{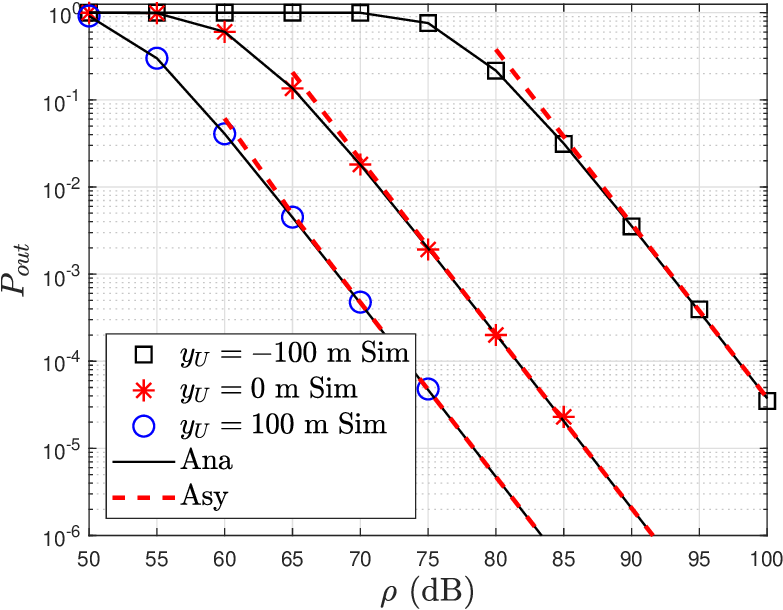}}
\caption{{The effect of UAV's position and $\rho$ on the OP of $D_F$ with ${\Theta _{th}} > 1 + \frac{1}{{{\Theta _B} - 1}}$}.}
\label{fig03}
\end{figure}
Fig. \ref{fig03} demonstrates the effect of UAV's position and the transmission SNR on the OP of $D_F$ when ${\Theta _{th}} > 1 + \frac{1}{{{\Theta _B} - 1}}$. 
It can be observed from Fig. \ref{fig03a} and Fig. \ref{fig03c} that the OP deteriorates as $\rho$ decreases at lower-$\rho$ region.
However, there is a floor for OP, which denotes that the OP gradually approaches a constant at the higher-$\rho$ region.
This is because the base station needs to allocate more power to $D_B$, resulting in too much interference on $D_F$. 
When the signals for $D_B$ cannot be decoded on $D_F$ with increasing transmit SNR,  the SINR at $D_F$ tends to a constant, which is independent of $\rho $. 
As shown in Fig. \ref{fig03b} and Fig. \ref{fig03d}, one can also observe that the outage performance with the DPA scheme outperforms that with the FPA scheme. 
Moreover, it can be found that the outage performance of $D_F$ in the DPA scheme will also improve with the increase of $\rho$ in the case of  high-$\rho$ region. 
Specifically, the OP floor problem is solved by the DPA scheme.
The reason is that more power allocated to $D_B$ ensures that $D_F$ can decode the signal of $D_B$ successfully. 
With the SIC technology, the inter-user interference is deleted and the achievable rate of $D_F$ is improved, thereby the performance is enhanced. 

Furthermore, it can be seen in Figs. \ref{fig02} and \ref{fig03} that OP is improved as the distance between $U$ and $D_F$ decreases. 
The reason is given as follows. Compared with the probability of loss propagation, the main factor on OP is path loss, which decreases with distance.

\begin{figure}[t]
\centering
\subfigure[OP with FPA in suburban environments for varying $\rho$, ${R_{{\rm{th}}}^B}$, and ${R_{{\rm{th}}}^F}$.]{
\label{fig04a}
\includegraphics[width = 0.2285 \textwidth]{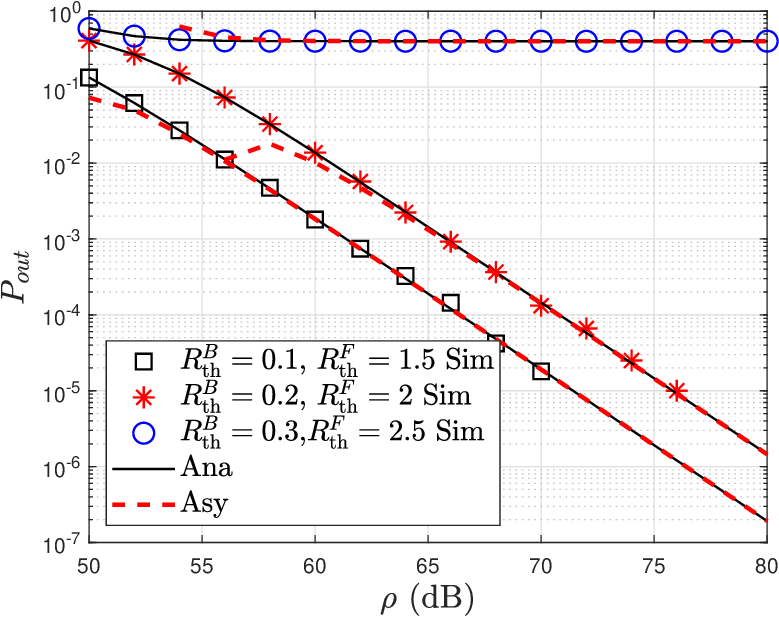}}
\subfigure[OP with DPA in suburban environments for varying $\rho$, ${R_{{\rm{th}}}^B}$, and ${R_{{\rm{th}}}^F}$.]{
\label{fig04b}
\includegraphics[width = 0.2285 \textwidth]{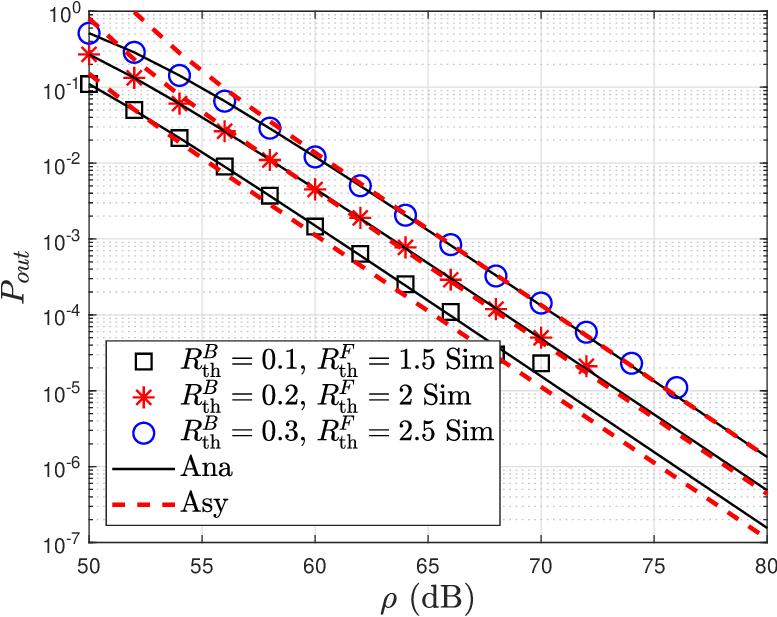}}
\subfigure[OP with FPA in urban environments for varying $\rho$, ${R_{{\rm{th}}}^B}$, and ${R_{{\rm{th}}}^F}$.]{
\label{fig04c}
\includegraphics[width = 0.2285 \textwidth]{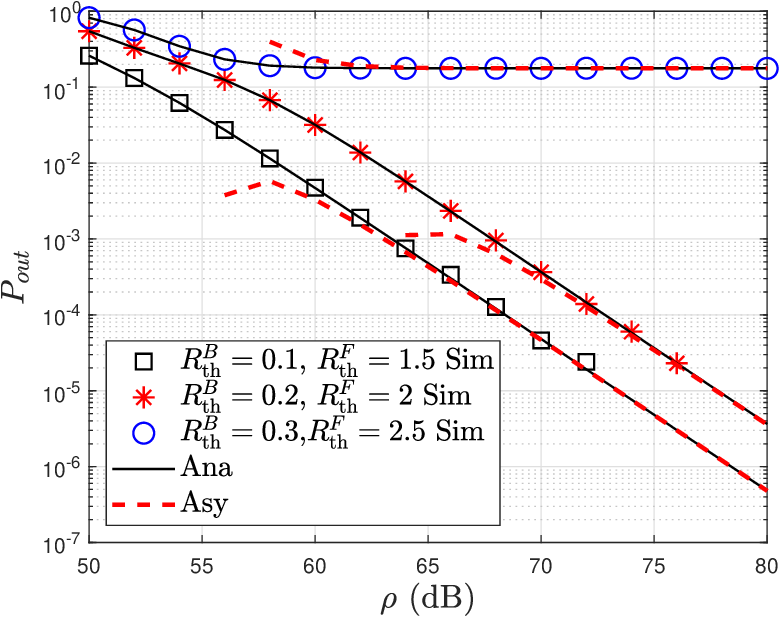}}
\subfigure[OP with DPA in urban environments for varying $\rho$, ${R_{{\rm{th}}}^B}$, and ${R_{{\rm{th}}}^F}$.]{
\label{fig04d}
\includegraphics[width = 0.2285 \textwidth]{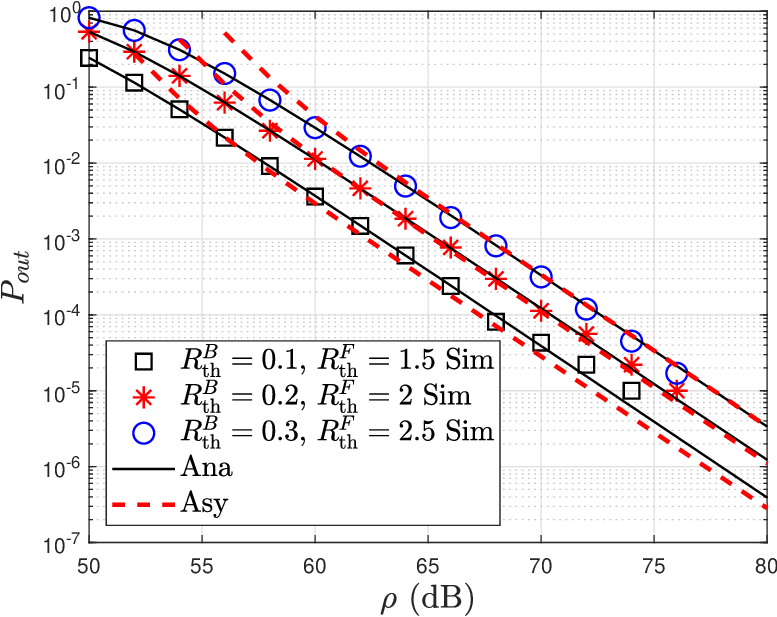}}
\caption{{The impact of RRT of $D_B$ and $D_F$ on the OP of $D_F$}.}
\label{fig04}
\end{figure}
Fig. \ref{fig04} demonstrates the impact of RRT of $D_B$ and $D_F$ on the OP of $D_F$. 
It can be observed that the larger the ${R_{{\rm{th}}}^B}$ and ${R_{{\rm{th}}}^F}$, the larger the OP, which is easy to follow because larger RRT denotes higher requirement.
As demonstrated in Fig. \ref{fig04a} and Fig. \ref{fig04c}, it can be observed that the relationship between the RRT for $D_B$ and $D_F$ under the condition in which the resources can be shared also significantly affects the OP of $D_F$, which is testified in Theorem 1. 
Furthermore, the results in Fig. \ref{fig04b} and Fig. \ref{fig04d} verify that the DPA scheme solves the OP floors perfectly. 
Then, the effectiveness of the DPA scheme is testified. 
It must be noted that the power allocation in the DPA scheme not only depends on the global CSI but also on the RRT of $D_B$, which is expressed in Eq. (\ref{rateDPA}). 
Thus, jointly designing the RRT of $D_B$ and $D_F$ based on the global CSI can maximize the achievable rate of GF users while ensuring the QoS of the GB user, which will be part of future work.


\begin{figure}[t]
\centering
\subfigure[OP with FPA in suburban environments for varying $y_U$ and $z_U$.]{
\label{fig05a}
\includegraphics[width = 0.2285 \textwidth]{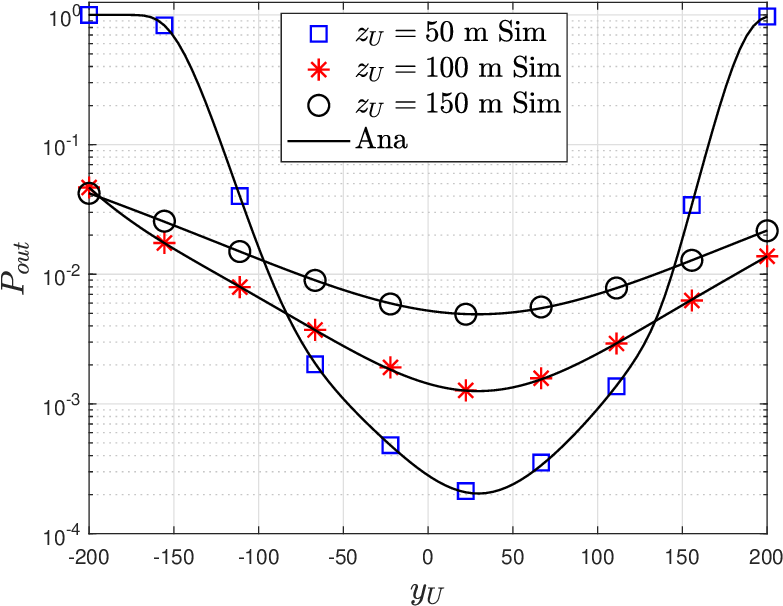}}
\subfigure[OP with DPA in suburban environments for varying $y_U$ and $z_U$.]{
\label{fig05b}
\includegraphics[width = 0.2285 \textwidth]{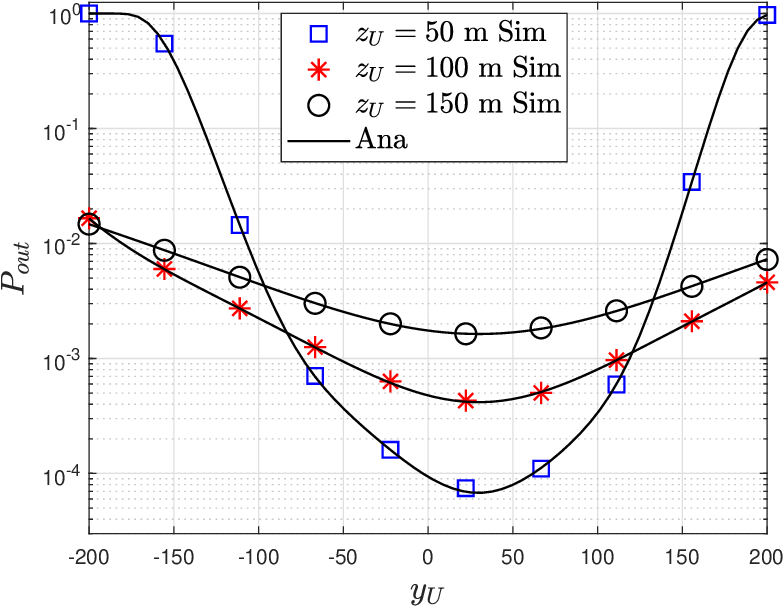}}
\subfigure[OP with FPA in urban environments for varying $y_U$ and $z_U$.]{
\label{fig05c}
\includegraphics[width = 0.2285 \textwidth]{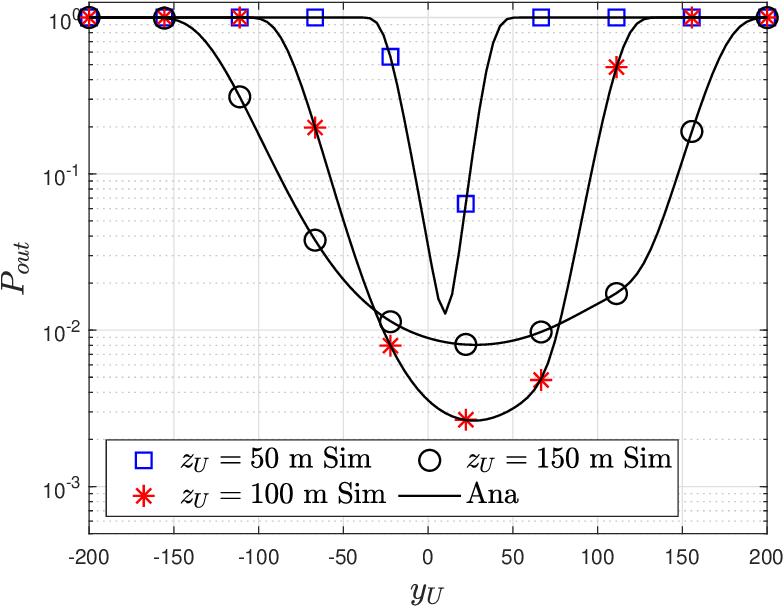}}
\subfigure[OP with DPA in urban environments for varying $y_U$ and $z_U$.]{
\label{fig05d}
\includegraphics[width = 0.2285 \textwidth]{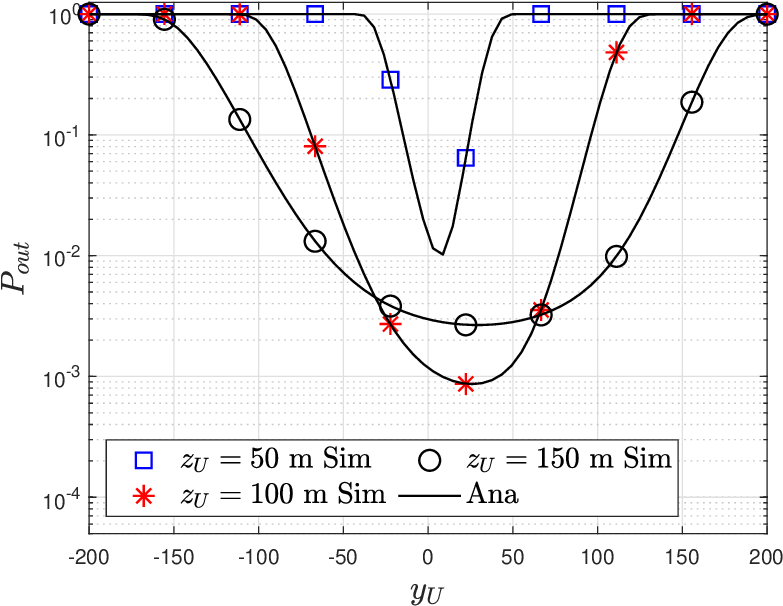}}
\caption{{The impact of the UAV's position and altitude on the OP of $D_F$ with ${R_{{\rm{th}}}^B} = 0.2$ and ${R_{{\rm{th}}}^F} = 2$}.}
\label{fig05}
\end{figure}

Figs. \ref{fig05} and \ref{fig06} demonstrate the impact of the UAV's position and altitude on the OP of the GF user with varying RRT requirements.
One can observe that the OP first decreases and then increases, which means there is an optimal location for the UAV so that the GF user can obtain the optimal performance. 
At the same time, it can be observed that under different environments and power allocation schemes, the 3D location of the UAV has different effects on the OP. 
This is because the trade-off between the probability of LoS propagation and the path loss caused by long-distance communication are different.
Comparing Fig. \ref{fig05} and Fig. \ref{fig06}, it can be observed that the OP with the DPA scheme outperforms that with the FPA scheme. 
The result in Fig. \ref{fig05a} and Fig. \ref{fig06a} demonstrates that the relationship between the RRT for $D_B$ and $D_F$ makes a big difference to the OP of $D_F$ in suburban environments with the FPA scheme. 
The same conclusion can also be observed by comparing Fig. \ref{fig05c} and Fig. \ref{fig06c}. 
However, the effect of the relationship between the RRT for $D_B$ and $D_F$ makes little difference to the OP of $D_F$ in the same environments with the DPA scheme, which is found by comparing Fig. \ref{fig05b} and Fig. \ref{fig06b}, and Fig. \ref{fig05d} and Fig. \ref{fig06d}. 
The reason is that the CSI of $D_F$ is also considered in the DPA scheme, which is expressed in Eq. (\ref{rateDPA}).

\begin{figure}[t]
\centering
\subfigure[OP with FPA in suburban environments for varying $y_U$ and $z_U$.]{
\label{fig06a}
\includegraphics[width = 0.2285 \textwidth]{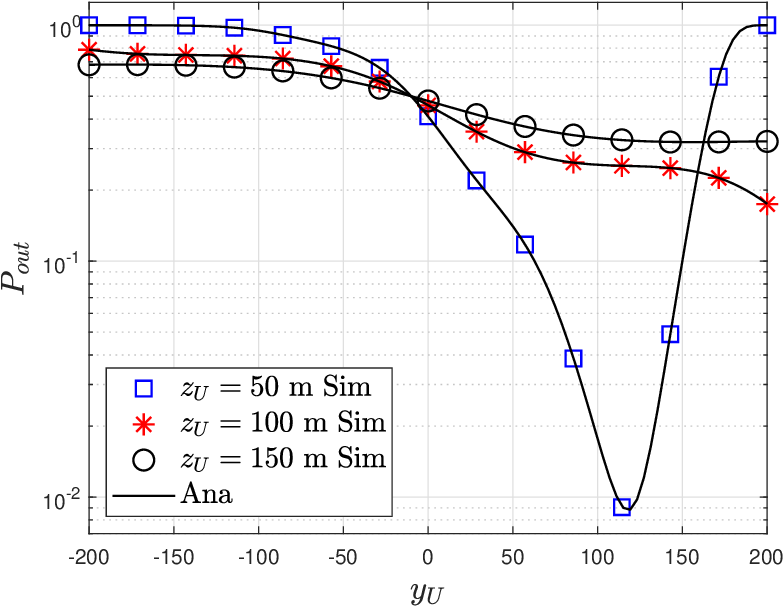}}
\subfigure[OP with DPA in suburban environments for varying $y_U$ and $z_U$.]{
\label{fig06b}
\includegraphics[width = 0.2285 \textwidth]{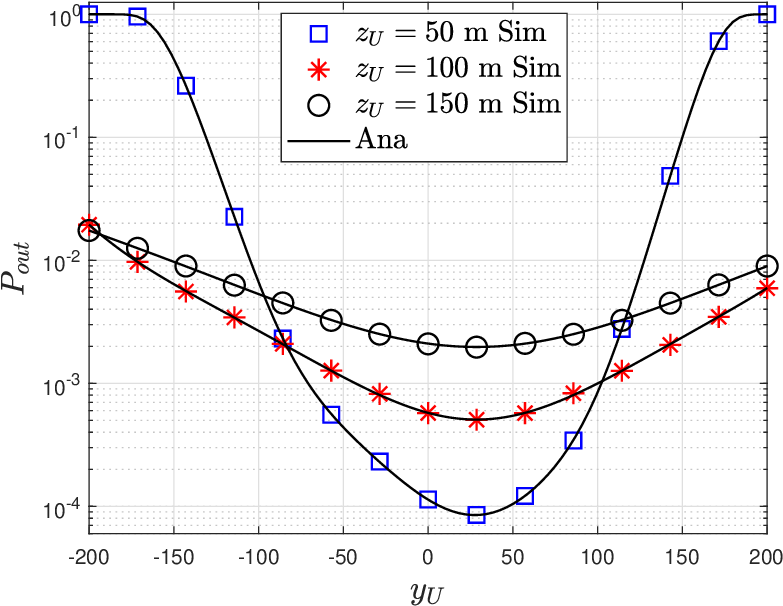}}
\subfigure[OP with FPA in urban environments for varying $y_U$ and $z_U$.]{
\label{fig06c}
\includegraphics[width = 0.2285 \textwidth]{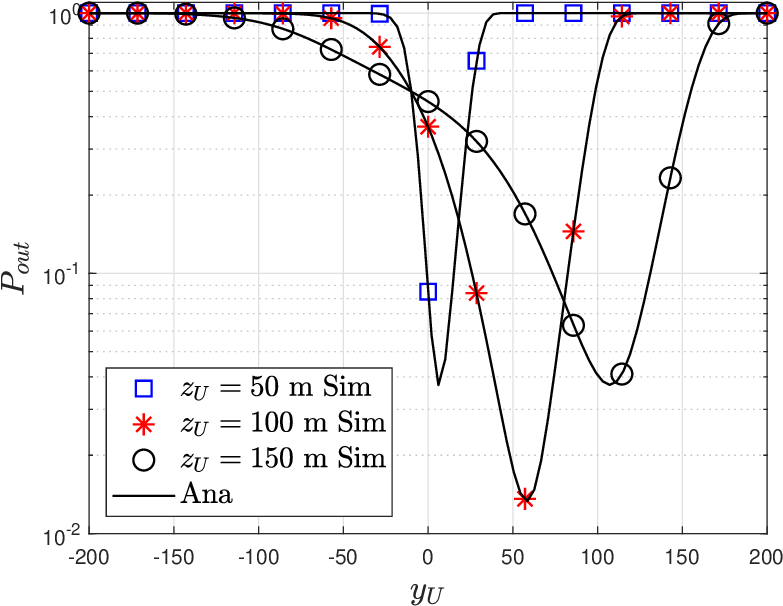}}
\subfigure[OP with DPA in urban environments for varying $y_U$ and $z_U$.]{
\label{fig06d}
\includegraphics[width = 0.2285 \textwidth]{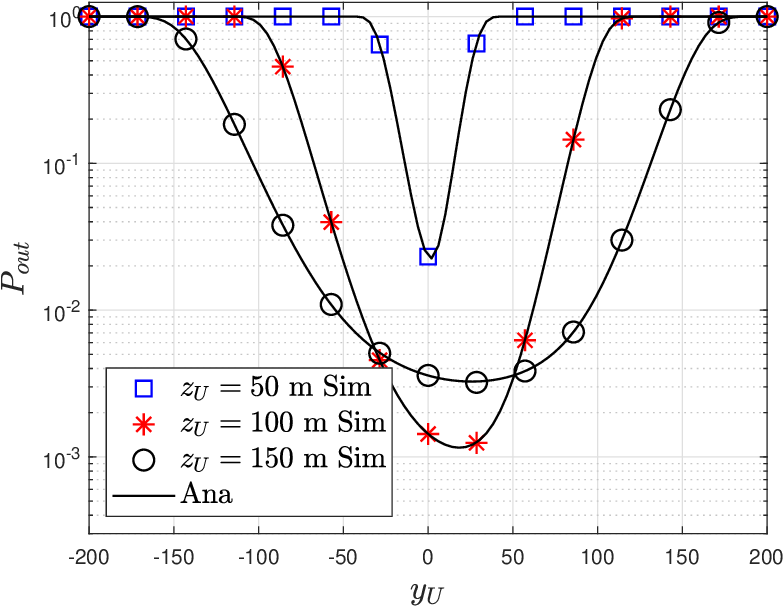}}
\caption{ {The impact of the UAV's position and altitude on the OP of $D_F$ with ${R_{{\rm{th}}}^B} = 0.5$ and ${R_{{\rm{th}}}^F} = 2.5$.}}
\label{fig06}
\end{figure}

\section{Conclusions}
\label{sec:Conclusion}

This work analyzed the outage performance of an aerial SGF NOMA system. 
Firstly, the outage performance of UAV-enabled downlink NOMA systems was analyzed with the SGF transmission scheme.
The exact and asymptotic expressions for the OP of the GF user under the FPA scheme were derived.
It was found that there are OP floors under stringent conditions on quality of service requirements.
A DPA scheme was proposed to eliminate the OP floor at the high-SNR region. The analytical expressions for the exact and asymptotic OP of the GF user were derived applicable to this later scheme.
Numerical and simulation results demonstrate that the outage performance of the GF user was improved by utilizing the proposed DPA scheme.
And the effects of system parameters, such as UAV location and altitude, on the outage performance were analyzed.
It was observed that there is an optimal position for the UAV to minimize $D_F$'s OP. The optimal placement is related to the environment, the altitude of the UAV, and the transmission SNR. Optimizing the outage performance of the GF user by designing the trajectory of the UAV will be conducted as part of our future work. 

\section*{Acknowledgements}
This work was supported by the National Natural Science Foundation of China under Grant 61971080 and the Open Fund of the Shaanxi Key Laboratory of Information Communication Network and Security under Grant ICNS201807. 
The authors would like to thank the editor for efficiently handling the review of this paper and the anonymous reviewers for their valuable suggestions and critical comments that helped to improve the quality of the paper.

\begin{appendices}
\section{Proof of Theorem 1}
\label{prooftheorem1}

Based on (\ref{G_X_cdf}), we obtain ${T_0} = {F_{{G_B}}}\left( {{\varepsilon _1}} \right)$.

{On the ground of} (\ref{rateF}), $T_1$ is expressed as
\begin{equation}
\begin{aligned}
{T_1} &= \underbrace {\Pr \left\{ {{G_B} > {\varepsilon _1},{G_F} > {G_B},R_F^1 < {R_{{\rm{th}}}^F} } \right\}}_{{T_{11}}}\\
&+ \underbrace {\Pr \left\{ {{G_B} > {\varepsilon _1},{G_F} < {G_B},R_F^2 < {R_{{\rm{th}}}^F}} \right\}}_{{T_{12}}}
\label{T101}
\end{aligned}
\end{equation}
where
${T_{11}}$ denotes $D_F$'s signal is decoded at the first stage of SIC and ${T_{12}}$ denotes $D_F$'s signal is decoded at the second stage of SIC.
{From} (\ref{pacDB}) and (\ref{rateF1}), we have
\begin{equation}
\begin{aligned}
&\Pr \left\{ {{{\log }_2}\left( {1 + {\bar \omega }\rho {G_F}} \right) < {R_{{\rm{th}}}^F} } \right\} \\
&= \Pr \left\{ {{G_F} < \frac{{{\Theta _{th}} - 1}}{{{\bar \omega }\rho }} } \right\}\\
&= \Pr \left\{ {{G_F} < \frac{{{\varepsilon _2}{G_B}}}{{{G_B} - {\varepsilon _1}}} } \right\}
\label{eq24}
\end{aligned}
\end{equation}
where
${\Theta _{th}} = {2^{R_{{\rm{th}}}^F}}$,
${\varepsilon _2} = \frac{{{\Theta _B}\left( {{\Theta _{th}} - 1} \right)}}{\rho }$,
and
${\Theta _B} = {2^{R_{{\rm{th}}}^B}}$.
Then, {grounded} on (\ref{G_X_cdf}) and (\ref{G_X_pdf}), $T_{11}$ is expressed as 
\begin{equation}
\begin{aligned}
{T_{11}} &= \Pr \left\{ {{G_B} > {\varepsilon _1},{G_F} > {G_B},{G_F} < \frac{{{\varepsilon _2}{G_B}}}{{{G_B} - {\varepsilon _1}}} } \right\}\\
& = \Pr \left\{ {{G_B} > {\varepsilon _1},{G_B} < {G_F} < \frac{{{\varepsilon _2}{G_B}}}{{{G_B} - {\varepsilon _1}}},} \right.\\
&\;\;\;\;\;\;\;\;\;\;\; \left. {{G_B} < \frac{{{\varepsilon _2}{G_B}}}{{{G_B} - {\varepsilon _1}}}} \right\}\\
&= \Pr \left\{ {{G_B} < {G_F} < \frac{{{\varepsilon _2}{G_B}}}{{{G_B} - {\varepsilon _1}}},{\varepsilon _1} < {G_B} < {\varepsilon _0} } \right\}\\
&= \underbrace {\Pr \left\{ {{G_F} < \frac{{{\varepsilon _2}{G_B}}}{{{G_B} - {\varepsilon _1}}},{\varepsilon _1} < {G_B} < {\varepsilon _0}{\mkern 1mu} } \right\}}_{ \buildrel \Delta \over = {\chi _1}} \\
&\;\;\;\;\;\; - \underbrace {\Pr \left\{ {{G_B} < {G_F},{\varepsilon _1} < {G_B} < {\varepsilon _0}{\mkern 1mu} } \right\}}_{ \buildrel \Delta \over = {\chi _2}}
\label{T1101}
\end{aligned}
\end{equation}	
where
${\varepsilon _0} = {\varepsilon _1} + {\varepsilon _2}  = \frac{{{\Theta _B}{\Theta _{th}} - 1}}{\rho }$ and 
${\chi _1}$ is expressed as 
\begin{equation}
\begin{aligned}
{\chi _1} &= \Pr \left\{ {{G_F} < \frac{{{\varepsilon _2}{G_B}}}{{{G_B} - {\varepsilon _1}}},{\varepsilon _1} < {G_B} < {\varepsilon _0}{\mkern 1mu} } \right\}\\
&= \int_{{\varepsilon _1}}^{{\varepsilon _0}} {{f_{{G_B}}}\left( x \right){F_{{G_F}}}\left( {\frac{{{\varepsilon _2}x}}{{x - {\varepsilon _1}}}} \right)dx} \\
&= {F_{{G_B}}}\left( {{\varepsilon _0}} \right) - {F_{{G_B}}}\left( {{\varepsilon _1}} \right) - {A_1}{\Phi_1}
\label{chi1}
\end{aligned}
\end{equation}
where
${A_1} = \frac{{\lambda _{B}^m}}{{\Gamma \left( m \right)}}$
and
$\Phi_1  = \sum\limits_{i = 0}^{m - 1} {\frac{{\lambda _F^i}}{{i!}}\int_{{\varepsilon _1}}^{{\varepsilon _0}} {{y^{m - 1}}{{\left( {\frac{{{\varepsilon _2}y}}{{y - {\varepsilon _1}}}} \right)}^i}{e^{ - {\lambda _B}y - \frac{{{\lambda _F}{\varepsilon _2}y}}{{y - {\varepsilon _1}}}}}dy} } $.
To the authors' best knowledge, it is very difficult to obtain the closed-form expression of $\Phi _1$.
To facilitate the following analysis, we define
\begin{equation}
\begin{aligned}
{g_1}\left( {a,b,s,t} \right) &= \sum\limits_{i = 0}^{m - 1} {\frac{{{{\left( {{\lambda _F}b} \right)}^i}}}{{i!}}} \int_s^t {\frac{{{y^{m + i - 1}}}}{{{{\left( {y - a} \right)}^i}}}{e^{ - {\lambda _B}y - \frac{{b{\lambda _F}y}}{{y - a}}}}dy}
\end{aligned}
\end{equation}
Utilizing Gaussian-Chebyshev quadrature \cite[(25.4.39)]{Abramowitz1972Book}, we obtain the approximation for ${g_1}\left( {a,b,s,t} \right)$ as
\begin{equation}
\begin{aligned}
{g_1}\left( {a,b,s,t} \right) &= \frac{\pi }{N}\sum\limits_{i = 0}^{m - 1} {\frac{{{{\left( {{\lambda _F}b} \right)}^i}}}{{i!}}\sum\limits_{n = 1}^N {\frac{{{{\left( {{\mu _n}\left( {s,t} \right)} \right)}^{m + i - 1}}}}{{{{\left( {{\mu _n}\left( {s,t} \right) - a} \right)}^i}}}} } \\
&\times {e^{ - {\lambda _B}{\mu _n}\left( {s,t} \right) - \frac{{{\lambda _F}b{\mu _n}\left( {s,t} \right)}}{{{\mu _n}\left( {s,t} \right) - a}}}}\\
&\times \sqrt {\left( {{\mu _n}\left( {s,t} \right) - s} \right)\left( {t - {\mu _n}\left( {s,t} \right)} \right)}
\label{g1}
\end{aligned}
\end{equation}
where ${\mu _n}\left( {s,t} \right) = \frac{{t + s}}{2} + \frac{{\left( {t - s} \right){\tau _n}}}{2}$, ${\tau _n} = \cos \frac{{\left( {2n - 1} \right)\pi }}{{2N}}$, and
$N$ is the summation terms, which reflects accuracy vs. complexity.
Then, we obtain ${\Phi _1} = {g_1}\left( {{\varepsilon _1},{\varepsilon _2},{\varepsilon _1},{\varepsilon _0}} \right)$.

Utilizing \cite[(3.351.1)]{GradshteynBook}, ${\chi _2}$ is obtained as
\begin{equation}
\begin{aligned}
{\chi _2} &= \Pr \left\{ {{G_B} < {G_F},{\varepsilon _1} < {G_B} < {\varepsilon _0}{\mkern 1mu} } \right\}\\
&= \int_{{\varepsilon _1}}^{{\varepsilon _0}} {{f_{{G_B}}}\left( y \right){F_{{G_F}}}\left( y \right)dy} \\
&= {F_{{G_B}}}\left( {{\varepsilon _0}} \right) - {F_{{G_B}}}\left( {{\varepsilon _1}} \right) - \frac{{\lambda _{{G_B}}^m}}{{\Gamma \left( m \right)}}\sum\limits_{i = 0}^{m - 1} {\frac{{\lambda _F^i}}{{i!}}} \\
&\int_{{\varepsilon _1}}^{{\varepsilon _0}} {{e^{ - \left( {{\lambda _B} + {\lambda _F}} \right)y}}{y^{m + i - 1}}dy} \\
& = {F_{{G_B}}}\left( {{\varepsilon _0}} \right) - {F_{{G_B}}}\left( {{\varepsilon _1}} \right) - {A_1}\sum\limits_{i = 0}^{m - 1} {\frac{{\lambda _F^i{\Phi _2}}}{{i!}}}
\label{chi2}
\end{aligned}
\end{equation}
where
${A_2} = {\lambda _{{B}}} + {\lambda _{{F}}}$,
${\Phi _2} = \frac{{\Upsilon \left( {i + m,{A_2}{\varepsilon _0}} \right) - \Upsilon \left( {i + m,{A_2}{\varepsilon _1}} \right)}}{{A_2^{i + m}}}$,
and
$\Upsilon \left( { \cdot , \cdot } \right)$ is lower incomplete Gamma function as defined by  \cite[(8.350.2)]{GradshteynBook}.

Similar to ${T_{11}}$, $T_{12}$ is expressed as
\begin{equation}
\begin{aligned}
{T_{12}} & = \Pr \left\{ {{G_B} > {\varepsilon _1},{G_F} < {G_B},} \right.\\
& \left. {{{\log }_2}\left( {1 + \frac{{\bar \omega \rho {G_F}}}{{1 + \rho \omega {G_F}}}} \right) < R_{{\rm{th}}}^F} \right\}\\
&  = \Pr \left\{ {{G_B} > {\varepsilon _1},{G_F} < {G_B},} \right.\\
& \left. {\rho \left( {1 - {\Theta _{th}}\omega } \right){G_F} < {\Theta _{th}} - 1} \right\}
\label{T201}
\end{aligned}
\end{equation}
Since there is
\begin{equation}
\begin{aligned}
&\Pr \left\{ {\rho \left( {1 - {\Theta _{th}}{\omega }} \right){G_F} < {\Theta _{th}} - 1} \right\} =  \Pr \left\{ {1 - {\Theta _{th}}{\omega } < 0} \right\}\\
&+\Pr \left\{ {{G_F} < \frac{{{\varepsilon _4}{G_B}}}{{{G_B} - {\varepsilon _3}}},1 - {\Theta _{th}}{\omega } > 0} \right\}
\label{eq31}
\end{aligned}
\end{equation}
where ${\varepsilon _3} = \frac{{{\Theta _B} - 1}}{\rho }\frac{{{\Theta _{th}}}}{{{\Theta _{th}} - \left( {{\Theta _{th}} - 1} \right){\Theta _B}}} = {\varepsilon _1}\frac{{{\Theta _{th}}}}{{{\Theta _{th}} - \left( {{\Theta _{th}} - 1} \right){\Theta _B}}} > {\varepsilon _1}$ and ${\varepsilon _4} = \frac{{{\Theta _B}\left( {{\Theta _{th}} - 1} \right)}}{{\left( {{\Theta _B} - {\Theta _{th}}\left( {{\Theta _B} - 1} \right)} \right)\rho }}$.
Furthermore, we obtain the following result 
\begin{equation}
\begin{aligned}
&\Pr \left\{ {1 - {\omega }{\Theta _{th}} < 0} \right\} \\
&= \Pr \left\{ {1 - \frac{{\rho {\Theta _B}{G_B} - \rho {G_B} + {\Theta _B} - 1}}{{\rho {\Theta _B}{G_B}}}{\Theta _{th}} < 0} \right\}\\
&= \Pr \left\{ {\rho \left( {\frac{{{\Theta _B}}}{{{\Theta _B} - 1}} - {\Theta _{th}}} \right){G_B} < {\Theta _{th}}} \right\}
\label{eq32}
\end{aligned}
\end{equation}
Thus, the relationship between ${{\Theta _{th}}}$ and ${\frac{{{\Theta _B}}}{{{\Theta _B} - 1}}}$ must be considered first.

1) When ${\Theta _{th}} < \frac{{{\Theta _B}}}{{{\Theta _B} - 1}}$, we have
\begin{equation}
\begin{aligned}
\Pr \left\{ {1 - {\omega }{\Theta _{th}}} < 0 \right\} & = \Pr \left\{ {{G_B} < {\varepsilon _3}} \right\}
\label{A3}
\end{aligned}
\end{equation}
Thus, we have
\begin{equation}
\begin{aligned}
&\Pr \left\{ {\rho \left( {1 - {\Theta _{th}}{\omega }} \right){G_F} < {\Theta _{th}} - 1} \right\} \\
&= \Pr \left\{ {{G_B} < {\varepsilon _3}} \right\}  \\
& + \Pr \left\{ {{G_F} < \frac{{{\varepsilon _4}{G_B}}}{{{G_B} - {\varepsilon _3}}},{G_B} > {\varepsilon _3}} \right\}
\label{delat1}
\end{aligned}
\end{equation}
Substituting (\ref{delat1}) into (\ref{T201}), we have
\begin{equation}
\begin{aligned}
&{T_{12a}} = \Pr \left\{ {{G_F} < {G_B},{\varepsilon _1} < {G_B} < {\varepsilon _3}} \right\} \\
&+ \Pr \left\{ {{G_B} > {\varepsilon _3},{G_F} < \min \left( {{G_B},\frac{{{\varepsilon _4}{G_B}}}{{{G_B} - {\varepsilon _3}}}} \right)} \right\}	
\label{T12a1}
\end{aligned}
\end{equation}
The 2nd term in (\ref{T12a1}) is derived as
\begin{equation}
\begin{small}
\begin{aligned}
	&\Pr \left\{ {{G_B} > {\varepsilon _3},{G_F} < \min \left( {{G_B},\frac{{{\varepsilon _4}{G_B}}}{{{G_B} - {\varepsilon _3}}}} \right)} \right\}\\
	&= \Pr \left\{ {{G_B} > {\varepsilon _3},{G_F} < {G_B},{G_B} < \frac{{{\varepsilon _4}{G_B}}}{{{G_B} - {\varepsilon _3}}}} \right\} \\
	&+ \Pr \left\{ {{G_B} > {\varepsilon _3},{G_F} < \frac{{{\varepsilon _4}{G_B}}}{{{G_B} - {\varepsilon _3}}},{G_B} > \frac{{{\varepsilon _4}{G_B}}}{{{G_B} - {\varepsilon _3}}}} \right\}\\
	&= \Pr \left\{ {{G_B} > {\varepsilon _3},{G_F} < {G_B},{G_B} < {\varepsilon _5}} \right\} \\
	&+ \Pr \left\{ {{G_B} > {\varepsilon _3},{G_F} < \frac{{{\varepsilon _4}{G_B}}}{{{G_B} - {\varepsilon _3}}},{G_B} > {\varepsilon _5}} \right\}
	\label{eq36}
\end{aligned}
\end{small}
\end{equation}
where
${\varepsilon _5} = {\varepsilon _3} + {\varepsilon _4} = \frac{{2{\Theta _B}{\Theta _{th}} - {\Theta _{th}} - {\Theta _B}}}{{\rho \left( {{\Theta _{th}} + {\Theta _B} - {\Theta _{th}}{\Theta _B}} \right)}}$.	
Then, $T_{12a}$ is expressed as
\begin{equation}
\begin{aligned}
{T_{12a}} &= \Pr \left\{ {{G_F} < {G_B},{\varepsilon _1} < {G_B} < {\varepsilon _3}} \right\} \\
&+ \Pr \left\{ {{G_F} < {G_B},{\varepsilon _3} < {G_B} < {\varepsilon _5}} \right\}\\
&+ \Pr \left\{ {{G_F} < \frac{{{\varepsilon _4}{G_B}}}{{{G_B} - {\varepsilon _3}}},{G_B} > {\varepsilon _3},{G_B} > {\varepsilon _5}} \right\}\\
&= \underbrace {\Pr \left\{ {{G_F} < {G_B},{\varepsilon _1} < {G_B} < {\varepsilon _5}} \right\}}_{ \buildrel \Delta \over = {\chi _3}} \\
&+ \underbrace {\Pr \left\{ {{G_F} < \frac{{{\varepsilon _4}{G_B}}}{{{G_B} - {\varepsilon _3}}},{G_B} > {\varepsilon _5}} \right\}}_{ \buildrel \Delta \over = {\chi _4}}
\label{T12a2}
\end{aligned}
\end{equation}
With the same method as (\ref{chi2}), we obtain
\begin{equation}
{\chi _3} = {F_{{G_B}}}\left( {{\varepsilon _5}} \right) - {F_{{G_B}}}\left( {{\varepsilon _1}} \right) - {A_1}\sum\limits_{i = 0}^{m - 1} {\frac{{\lambda _F^i{\Phi _3}}}{{i!}}}
\label{chi3}
\end{equation}
where
${\Phi _3} = \frac{{\Upsilon \left( {i + m,{A_2}{\varepsilon _5}} \right) - \Upsilon \left( {i + m,{A_2}{\varepsilon _1}} \right)}}{{A_2^{i + m}}}$.
With the similar method as (\ref{chi1}), ${\chi _4}$ is obtained as
\begin{equation}
\begin{aligned}
{\chi _4} &= \Pr \left\{ {{G_F} < \frac{{{\varepsilon _4}{G_B}}}{{{G_B} - {\varepsilon _3}}},{G_B} > {\varepsilon _5}} \right\}\\
&= \int_{{\varepsilon _5}}^\infty  {{f_{{G_B}}}\left( y \right){F_{{G_F}}}\left( {\frac{{{\varepsilon _4}y}}{{y - {\varepsilon _3}}}} \right)dy} \\
&= {{\bar F}_{{G_B}}}\left( {{\varepsilon _5}} \right) - {A_1}{\Phi _4}
\label{chi4}
\end{aligned}
\end{equation}
where
${{\bar F}_{{G_X}}}\left( x \right) = 1 - {F_{{G_X}}}\left( x \right)$ and
${\Phi _4} = \sum\limits_{i = 0}^{m - 1} {\frac{{{{\left( {{\lambda _F}{\varepsilon _4}} \right)}^i}}}{{i!}}} \int_{{\varepsilon _5}}^\infty  {{e^{ - {\lambda _B}y - \frac{{{\lambda _F}{\varepsilon _4}y}}{{y - {\varepsilon _3}}}}}\frac{{{y^{m + i - 1}}}}{{{{\left( {y - {\varepsilon _3}} \right)}^i}}}dy}$.
To facilitate the following analysis, we define
\begin{equation}
\begin{aligned}
{g_2}\left( {a,b,c} \right) &= \sum\limits_{i = 0}^{m - 1} {\frac{{{{\left( {{\lambda _F}b} \right)}^i}}}{{i!}}} \int_c^\infty  {{e^{ - {\lambda _B}y - \frac{{b{\lambda _F}y}}{{y - a}}}}\frac{{{y^{m + i - 1}}}}{{{{\left( {y - a} \right)}^i}}}dy}
\label{g21}
\end{aligned}
\end{equation}
Utilizing Gaussian-Chebyshev quadrature \cite[(25.4.39)]{Abramowitz1972Book}, we obtain the approximation for ${g_2}\left( {a,b,c} \right)$ as
\begin{equation}
\begin{aligned}
&{g_2}\left( {a,b,c} \right) \\
&= \sum\limits_{i = 0}^{m - 1} {\frac{{{{\left( {{\lambda _F}b} \right)}^i}}}{{i!}}\sum\limits_{n = 1}^N {{w _n}{e^{{\iota _n} - \left( {{\lambda _B}{\iota _n} + \frac{{{\lambda _F}b{\iota _n}}}{{{\iota _n} - a}}} \right)}}\frac{{\iota _n^{m + i - 1}}}{{{{\left( {{\iota _n} - a} \right)}^i}}}} } \\
&- \sum\limits_{i = 0}^{m - 1} {\frac{{{{\left( {{\lambda _F}b} \right)}^i}\pi }}{{i!N}}\sum\limits_{n = 1}^N {{e^{ - \left( {{\lambda _B}{\mu _n}\left( {0,c} \right) + \frac{{{\lambda _F}b{\mu _n}\left( {0,c} \right)}}{{{\mu _n}\left( {0,c} \right) - a}}} \right)}}} } \\
&\times \frac{{{{\left( {{\mu _n}\left( {0,c} \right)} \right)}^{m + i - 1}}\sqrt {{\mu _n}\left( {0,c} \right)\left( {c - {\mu _n}\left( {0,c} \right)} \right)} }}{{{{\left( {{\mu _n}\left( {0,c} \right) - a} \right)}^i}}}
\label{g22}
\end{aligned}
\end{equation}
where $\iota _n$ is the $n$th zeros of Laguerre polynomials and ${w_n}$ is the Gaussian weight, which are given in Table (25.9) of \cite{Abramowitz1972Book}.
Then, we obtain ${\Phi _4} = {g_2}\left( {{\varepsilon _3},{\varepsilon _4},{\varepsilon _5}} \right)$.

2) When ${\Theta _{th}} > \frac{{{\Theta _B}}}{{{\Theta _B} - 1}}$, due to $\Pr \left\{ {1 - {\omega }{\Theta _{th}} < 0} \right\} = 1$, we obtain 
\begin{equation}
\begin{aligned}
{T_{12b}} & = \Pr \left\{ {{G_F} < {G_B},{G_B} > {\varepsilon _1}} \right\}\\
&= \int_{{\varepsilon _1}}^\infty  {{f_{{G_B}}}\left( y \right){F_{{G_F}}}\left( y \right)dy} \\
&= {{\bar F}_{{G_B}}}\left( {{\varepsilon _1}} \right) - \frac{{\lambda _B^m}}{{\Gamma \left( m \right)}}\sum\limits_{i = 0}^{m - 1} {\frac{{\lambda _F^i}}{{i!}}} \int_{{\varepsilon _1}}^\infty  {{y^{m + i - 1}}{e^{ - \left( {{\lambda _B} + {\lambda _F}} \right)y}}dy} \\
&= {{\bar F}_{{G_B}}}\left( {{\varepsilon _1}} \right) - {A_1}\sum\limits_{i = 0}^{m - 1} {\frac{{\lambda _F^i\Gamma \left( {i + m,{A_2}{\varepsilon _1}} \right)}}{{i!{A_2}^{i + m}}}}
\label{T12b}
\end{aligned}
\end{equation}
where $\Gamma \left( { \cdot , \cdot } \right)$ is upper incomplete Gamma function, which is defined by \cite[(3.351.2)]{GradshteynBook}.

\section{Proof of Corollary 1  }
\label{proofCorollary1}

Utilizing ${e^x} = \sum\limits_{j = 0}^\infty  {\frac{{{x^j}}}{{j!}}} $, we obtain $\sum\limits_{j = 0}^{n - 1} {\frac{{{x^j}}}{{j!}}}  = {e^x} - \frac{{{x^n}}}{{n!}} + O\left( {{x^n}} \right)$, then we have ${F_{{G_X}}}\left( x \right) \to \frac{{{{\left( {{\lambda _X}x} \right)}^m}}}{{m!}}$ when $x \to 0$.
Thus, we have $T_0^\infty  = F_{{G_B}}^\infty \left( {{\varepsilon _1}} \right) = \frac{{{{\left( {{\lambda _B}{\varepsilon _1}} \right)}^m}}}{{m!}}$.

When $\rho  \to \infty $, we have ${\varepsilon _1} \to 0$, ${\varepsilon _2} \to 0$, ${\varepsilon _3} \to 0$,
${\varepsilon _4} \to 0$,
and
$\frac{{{\varepsilon _2}{G_B} }}{{{G_B}  - {\varepsilon _1}}} \to {\varepsilon _2}$,
by utilizing $\Upsilon \left( {n,x} \right)\mathop  \approx  \limits^{x \to 0} \frac{{{x^n}}}{n}$, $T_{11}^\infty$ is obtained as 
\begin{equation}
\begin{aligned}
T_{11}^\infty  &= \Pr \left\{ {{G_B} < {G_F} < {\varepsilon _2},{\varepsilon _1} < {G_B} < {\varepsilon _0}} \right\}\\
&= \int_{{\varepsilon _1}}^{{\varepsilon _0}} {{f_{{G_B}}}\left( y \right)\left( {F_{{G_F}}^\infty \left( {{\varepsilon _2}} \right) - {F_{{G_F}}}\left( y \right)} \right)dy} \\
&= \int_{{\varepsilon _1}}^{{\varepsilon _0}} {{f_{{G_B}}}\left( y \right)\frac{{\lambda _F^m{\varepsilon _2}^m}}{{m!}}dy}  \\
&- \int_{{\varepsilon _1}}^{{\varepsilon _0}} {{f_{{G_B}}}\left( y \right)\left( {1 - {e^{ - {\lambda _F}y}}\sum\limits_{i = 0}^{m - 1} {\frac{{\lambda _F^i{y^i}}}{{i!}}} } \right)dy} \\
&= \left( {{\frac{{{{\left( {{\lambda _F}{\varepsilon _2}} \right)}^m}}}{{m!}} - 1}} \right)\int_{{\varepsilon _1}}^{{\varepsilon _0}} {{f_{{G_B}}}\left( y \right)dy}  \\
&+ \frac{{\lambda _B^m}}{{\Gamma \left( m \right)}}\sum\limits_{i = 0}^{m - 1} {\frac{{\lambda _F^i}}{{i!}}} \int_{{\varepsilon _1}}^{{\varepsilon _0}} {{y^{m + i - 1}}{e^{ - \left( {{\lambda _B} + {\lambda _F}} \right)y}}dy} \\
& \approx \frac{{\lambda _B^m\left( {\varepsilon _0^m - {\varepsilon _1}^m} \right)}}{{m!}}\left( {\frac{{{{\left( {{\lambda _F}{\varepsilon _2}} \right)}^m}}}{{m!}} - 1} \right) \\
&+  \frac{{\lambda _B^m}}{{\Gamma \left( m \right)}}\sum\limits_{i = 0}^{m - 1} {\frac{{\lambda _F^i\left( {\varepsilon _0^{i + m} - {\varepsilon _1}^{i + m}} \right)}}{{i!\left( {i + m} \right)}}}
\end{aligned}
\end{equation}

Similarly, due to
$\frac{{{\varepsilon _4}{G_B}}}{{{G_B} - {\varepsilon _3}}} \to {\varepsilon _4}$,
$T_{12a}^\infty$ is obtained as
\begin{equation}
T_{12a}^\infty  = \chi _3^\infty  + \chi _4^\infty
\end{equation}
where
\begin{equation}
\begin{aligned}
\chi _3^\infty &= \Pr \left\{ {{G_F} < {G_B},{\varepsilon _1} < {G_B} < {\varepsilon _5}} \right\}\\
&= \int_{{\varepsilon _1}}^{{\varepsilon _5}} {{f_{{G_B}}}\left( y \right){F_{{G_F}}}\left( y \right)dy} \\
&= \int_{{\varepsilon _1}}^{{\varepsilon _5}} {{f_{{G_B}}}\left( y \right)\left( {1 - {e^{ - {\lambda _F}y}}\sum\limits_{i = 0}^{m - 1} {\frac{{\lambda _F^i{y^i}}}{{i!}}} } \right)dy} \\
&\approx \frac{{\lambda _B^m\left( {\varepsilon _5^m - \varepsilon _1^m} \right)}}{{m!}} - {A_1}\sum\limits_{i = 0}^{m - 1} {\frac{{\lambda _F^i\left( {\varepsilon _5^{i + m} - \varepsilon _1^{i + m}} \right)}}{{i!\left( {i + m} \right)}}}
\label{chi3asy}
\end{aligned}
\end{equation}
and
\begin{equation}
\begin{aligned}
\chi _4^\infty & = \Pr \left\{ {{G_F} < {\varepsilon _4},{G_B} > {\varepsilon _5}} \right\}\\
&= \int_{{\varepsilon _5}}^\infty  {{f_B}\left( y \right)F_{{G_F}}^\infty \left( {{\varepsilon _4}} \right)dy} \\
&\approx \frac{{{{\left( {{\lambda _F}{\varepsilon _4}} \right)}^m}}}{{m!}}\left( {1 - \frac{{{{\left( {{\lambda _B}{\varepsilon _5}} \right)}^m}}}{{m!}}} \right)	
\label{chi4asy}
\end{aligned}
\end{equation}
Similarly, {on the ground of }(\ref{T12b}), $T_{12b}^\infty$ is obtained as
\begin{equation}
\begin{aligned}
T_{12b}^\infty  &\approx 1 - \frac{{{{\left( {{\lambda _B}{\varepsilon _1}} \right)}^m}}}{{m!}} \\
&- {A_1}\sum\limits_{i = 0}^{m - 1} {\frac{{\lambda _F^i}}{{i!{A_2}^{i + m}}}\left( {{\Gamma{\left( {i + m} \right)}} - \frac{{{{\left( {{A_2}{\varepsilon _1}} \right)}^{i + m}}}}{{i + m}}} \right)}
\label{A8}
\end{aligned}
\end{equation}

\section{Proof of Theorem 2  }
\label{prooftheorem2}

$T_0$ and $T_{11}$ are given in Theorem \ref{theorem1}.

Substituting (\ref{rateF2}) into (\ref{opDPA}), $T_2$ is expressed as
\begin{equation}
\begin{aligned}
{T_2} &= \Pr \left\{ {{G_B} > {\varepsilon _1},{G_F} < \frac{{{\Theta _B}{G_B}}}{{\rho {G_B} + 1}},} \right.\\
&\left. {\left( {1 - {\Theta _{th}}\omega } \right)\rho {G_F} < {\Theta _{th}} - 1} \right\}
\label{T2001}
\end{aligned}
\end{equation}

1) When ${\Theta _{th}} < \frac{{{\Theta _B}}}{{{\Theta _B} - 1}}$, we obtain 
\begin{equation}
\begin{footnotesize}
\begin{aligned}
	{T_{2a}} &= \Pr \left\{ {{G_B} > {\varepsilon _1},{G_F} < \frac{{{\Theta _B}{G_B}}}{{\rho {G_B} + 1}},{G_B} < {\varepsilon _3}} \right\}\\
	&+ \Pr \left\{ {{G_B} > {\varepsilon _1},{G_F} < \frac{{{\Theta _B}{G_B}}}{{\rho {G_B} + 1}},} \right.\\
	\;\;\;\;\;\;&\left. {{G_F} < \frac{{{\varepsilon _4}{G_B}}}{{{G_B} - {\varepsilon _3}}},{G_B} > {\varepsilon _3}} \right\}\\
	&= \underbrace {\Pr \left\{ {{G_F} < \frac{{{\Theta _B}{G_B}}}{{\rho {G_B} + 1}},{\varepsilon _1} < {G_B} < {\varepsilon _3}} \right\}}_{ \buildrel \Delta \over = {\chi _5}}\\
	&+ \underbrace {\Pr \left\{ {{G_F} < \min \left\{ {\frac{{{\Theta _B}{G_B}}}{{\rho {G_B} + 1}},\frac{{{\varepsilon _4}{G_B}}}{{{G_B} - {\varepsilon _3}}}} \right\},{G_B} > {\varepsilon _3}} \right\}}_{ \buildrel \Delta \over = {\chi _6}}	
\end{aligned}
\end{footnotesize}
\end{equation}
Considering the relationship between ${\frac{{{\Theta _B}{G_B}}}{{\rho {G_B} + 1}}}$ and ${\frac{{{\varepsilon _4}{G_B}}}{{{G_B} - {\varepsilon _3}}}}$, we obtain
\begin{equation}
\begin{aligned}
&\Pr \left\{ {\frac{{{\varepsilon _4}{G_B}}}{{{G_B} - {\varepsilon _3}}} - \frac{{{\Theta _B}{G_B}}}{{\rho {G_B} + 1}} > 0} \right\} \\
&= \left\{ {\begin{array}{*{20}{c}}
		{1,}&{{\Theta _B} > \frac{1}{{{\Theta _{th}} - 1}},{G_B} > {\varepsilon _3}}\\
		{\Pr \left\{ {{G_B} < {\varepsilon _6}} \right\},}&{{\Theta _B} < \frac{1}{{{\Theta _{th}} - 1}},{G_B} > {\varepsilon _3}}
\end{array}} \right.
\label{relation2}
\end{aligned}
\end{equation}
where ${\varepsilon _6} = \frac{{{\Theta _B}{\varepsilon _3} + {\varepsilon _4}}}{{{\Theta _B}\left( {{\Theta _B} + 1 - {\Theta _{th}}{\Theta _B}} \right)}}$.
Then, we obtain ${\chi _6} = \Pr \left\{ {{G_F} < \frac{{{\Theta _B}{G_B}}}{{\rho {G_B} + 1}},{G_B} > {\varepsilon _3}} \right\}$ for ${\Theta _B} > \frac{1}{{{\Theta _{th}} - 1}}$. Thus, we obtain
\begin{equation}
\begin{aligned}
T_{2a}^a &= \Pr \left\{ {{G_F} < \frac{{{\Theta _B}{G_B}}}{{\rho {G_B} + 1}},{\varepsilon _1} < {G_B} < {\varepsilon _3}} \right\}\\
&+ \Pr \left\{ {{G_F} < \frac{{{\Theta _B}{G_B}}}{{\rho {G_B} + 1}},{G_B} > {\varepsilon _3}} \right\}\\
&= \Pr \left\{ {{G_F} < \frac{{{\Theta _B}{G_B}}}{{\rho {G_B} + 1}},{G_B} > {\varepsilon _1}} \right\}\\
&= \int_{{\varepsilon _1}}^\infty  {{f_{{G_B}}}\left( y \right){F_{{G_F}}}\left( {\frac{{{\Theta _B}y}}{{\rho y + 1}}} \right)dy} \\
&= {{\bar F}_{{G_B}}}\left( {{\varepsilon _1}} \right) - {A_1}{g_2}\left( { - \frac{1}{\rho },\frac{{{\Theta _B}}}{\rho },{\varepsilon _1}} \right)
\label{eq51}
\end{aligned}
\end{equation}

For ${\Theta _B} < \frac{1}{{{\Theta _{th}} - 1}}$, we have
\begin{equation}
\begin{aligned}
{\chi _6} 
&= \Pr \left\{ {{G_F} < \frac{{{\Theta _B}{G_B}}}{{\rho {G_B} + 1}},{G_B} > {\varepsilon _3},{G_B} < {\varepsilon _6}} \right\}\\
&+ \Pr \left\{ {{G_F} < \frac{{{\varepsilon _4}{G_B}}}{{{G_B} - {\varepsilon _3}}},{G_B} > {\varepsilon _3},{G_B} > {\varepsilon _6}} \right\}
\label{chi652}
\end{aligned}
\end{equation}
Due to ${\varepsilon _6} - {\varepsilon _3} = \frac{{\left( {{\Theta _{th}} - 1} \right){\Theta _B}{\Theta _B}{\varepsilon _3} + {\varepsilon _4}}}{{{\Theta _B}\left( {1 - \left( {{\Theta _{th}} - 1} \right){\Theta _B}} \right)}} > 0$, we obtain (\ref{H23212}), shown at the top of the the next page.
\setcounter{equation}{53} 

2) When ${\Theta _{th}} > \frac{{{\Theta _B}}}{{{\Theta _B} - 1}}$, we have
\begin{equation}
{T_{2b}} = \Pr \left\{ {{G_F} < \frac{{{\Theta _B}{G_B}}}{{\rho {G_B} + 1}},{G_B} > {\varepsilon _1}} \right\} = {T_{2a}^a}
\end{equation}

With the similar method, $T_3$ is obtained as (\ref{T3001}), shown at the top of the next page,
where
\setcounter{equation}{55} 
${\Phi _5} = {g_2}\left( { - \frac{1}{\rho },\frac{{{\Theta _B}}}{\rho },{\varepsilon _1}} \right)$.

\setcounter{equation}{52}
\begin{figure*}[ht]
\begin{equation}
\begin{aligned}
	T_{2a}^b &=\Pr \left\{ {{G_F} < \frac{{{\Theta _B}{G_B}}}{{\rho {G_B} + 1}},{\varepsilon _1} < {G_B} < {\varepsilon _6}} \right\} + \Pr \left\{ {{G_F} < \frac{{{\varepsilon _4}{G_B}}}{{{G_B} - {\varepsilon _3}}},{G_B} > {\varepsilon _6}} \right\} \\
	& = \int_{{\varepsilon _1}}^{{\varepsilon _6}} {{f_{{G_B}}}\left( y \right){F_{{G_F}}}\left( {\frac{{{\Theta _B}y}}{{\rho y + 1}}} \right)dy}  + \int_{{\varepsilon _6}}^\infty  {{f_{{G_B}}}\left( y \right){F_{{G_F}}}\left( {\frac{{{\varepsilon _4}y}}{{y - {\varepsilon _3}}}} \right)dy}  \\
	&= \int_{{\varepsilon _1}}^\infty  {{f_{{G_B}}}\left( y \right)dy}  - \frac{{\lambda _B^m}}{{\Gamma \left( m \right)}}\sum\limits_{i = 0}^{m - 1} {\frac{{{{\left( {{\lambda _F}{\Theta _B}} \right)}^i}}}{{{\rho ^i}i!}}} \int_{{\varepsilon _1}}^{{\varepsilon _6}} {{y^{m - 1}}{e^{ - {\lambda _B}y - \frac{{{\lambda _F}\frac{{{\Theta _B}}}{\rho }y}}{{y + \frac{1}{\rho }}}}}\frac{{{y^{m + i - 1}}}}{{{{\left( {y + \frac{1}{\rho }} \right)}^i}}}dy}  \\
	&- \frac{{\lambda _B^m}}{{\Gamma \left( m \right)}}\sum\limits_{i = 0}^{m - 1} {\frac{{{{\left( {{\lambda _F}{\varepsilon _4}} \right)}^i}}}{{i!}}} \int_{{\varepsilon _6}}^\infty  {{e^{ - {\lambda _B}y - \frac{{{\lambda _F}{\varepsilon _4}y}}{{y - {\varepsilon _3}}}}}\frac{{{y^{m + i - 1}}}}{{{{\left( {y - {\varepsilon _3}} \right)}^i}}}dy}  \\
	&= {{\bar F}_{{G_B}}}\left( {{\varepsilon _1}} \right) - {A_1}{g_1}\left( { - \frac{1}{\rho },\frac{{{\Theta _B}}}{\rho },{\varepsilon _1},{\varepsilon _6}} \right) - {A_1}{g_2}\left( {{\varepsilon _3},{\varepsilon _4},{\varepsilon _6}} \right)
	\label{H23212}
\end{aligned}
\end{equation}
\hrulefill
\end{figure*}

\setcounter{equation}{54} 
\begin{figure*}[ht]
\begin{equation}
\begin{aligned}
	{T_3} &= \Pr \left\{ {{G_B} > {\varepsilon _1},\frac{{{\Theta _B}{G_B}}}{{\rho {G_B} + 1}} < {G_F} < {G_B},{G_F} < {\varepsilon _0}} \right\}\\
	&= \Pr \left\{ {\frac{{{\Theta _B}{G_B}}}{{\rho {G_B} + 1}} < {G_F} < {G_B},{\varepsilon _1} < {G_B} < {\varepsilon _0}} \right\} + \Pr \left\{ {\frac{{{\Theta _B}{G_B}}}{{\rho {G_B} + 1}} < {G_F} < {\varepsilon _0},{G_B} > {\varepsilon _0}} \right\}\\
	&= {F_{{G_F}}}\left( {{\varepsilon _0}} \right)\int_{{\varepsilon _0}}^\infty  {{f_{{G_B}}}\left( y \right)dy}  + \int_{{\varepsilon _1}}^{{\varepsilon _0}} {{f_{{G_B}}}\left( y \right){F_{{G_F}}}\left( y \right)dy}  - \int_{{\varepsilon _1}}^\infty  {{f_{{G_B}}}\left( y \right){F_{{G_F}}}\left( {\frac{{{\Theta _B}{y}}}{{\rho {y} + 1}}} \right)dy} \\
	&= {F_{{G_F}}}\left( {{\varepsilon _0}} \right){{\bar F}_{{G_B}}}\left( {{\varepsilon _0}} \right) + \int_{{\varepsilon _1}}^{{\varepsilon _0}} {{f_{{G_B}}}\left( y \right){F_{{G_F}}}\left( y \right)dy}  - \left( {1 - {F_{{G_B}}}\left( {{\varepsilon _1}} \right) - \frac{{\lambda _B^m}}{{\Gamma \left( m \right)}}\sum\limits_{i = 0}^{m - 1} {\frac{{{{\left( {{\lambda _F}} \right)}^i}{\Phi _2}}}{{i!}}} } \right) \hfill \\
	&= {A_1}{\Phi _5} - {{\bar F}_{{G_F}}}\left( {{\varepsilon _0}} \right){{\bar F}_{{G_B}}}\left( {{\varepsilon _0}} \right) - {A_1}\sum\limits_{i = 0}^{m - 1} {\frac{{{{\left( {{\lambda _F}} \right)}^i}{\Phi _2}}}{{i!}}}	
	\label{T3001}	
\end{aligned}
\end{equation}
\hrulefill
\end{figure*}

\section{Proof of Corollary 2  }
\label{proofCorollary2}

When $\rho  \to \infty $, we have $F_{{G_X}}^\infty \left( x \right) \to \frac{{{{\left( {{\lambda _X}x} \right)}^m}}}{{m!}}$, ${\varepsilon _1} \to 0$, ${\varepsilon _2} \to 0$, ${\varepsilon _3} \to 0$,
$\frac{{{\Theta _B}{G_B}}}{{\rho {G_B} + 1}} \to \frac{{{\Theta _B}}}{\rho }$, and
${\omega } \to 1 - \frac{1}{{{\Theta _B}}}$. For ${\Theta _{th}} > \frac{{{\Theta _B}}}{{{\Theta _B} - 1}}$, we have
\begin{equation}
\begin{aligned}
&\Pr \left\{ {R_F^2 < {R_{{\rm{th}}}^F}} \right\} \\
&= \Pr \left\{ {\log \left( {1 + \frac{{\rho {\bar \omega }{G_F}}}{{1 + \rho {\omega }{G_F}}}} \right) < {R_{{\rm{th}}}^F}} \right\}\\
&\approx 1
\label{eq56}
\end{aligned}
\end{equation}
Thus, 
$T_{2a}^\infty$ and $T_{2b}^\infty$ are obtained as
\begin{equation}
T_{2a}^{\infty } = 0
\label{eq57}
\end{equation}
and
\begin{equation}
\begin{aligned}
T_{2b}^{\infty } &= \Pr \left\{ {{G_B} > {\varepsilon _1},{G_F} < \frac{{{\Theta _B}}}{\rho }} \right\}\\
&= \frac{1}{{m!}}{\left( {\frac{{{\lambda _F}{\Theta _B}}}{\rho }} \right)^m}\left( {1 - \frac{{{{\left( {{\varepsilon _1}{\lambda _B}} \right)}^m}}}{{m!}}} \right)
\end{aligned}
\end{equation}
respectively.

Similarly, 
$T_3^\infty$ is obtained as
\begin{equation}
\begin{aligned}
T_3^\infty  &= \Pr \left\{ {\frac{{{\Theta _B}}}{\rho } < {G_F} < {G_B},{\varepsilon _1} < {G_B} < {\varepsilon _0} } \right\} \\
&+ \Pr \left\{ {\frac{{{\Theta _B}}}{\rho } < {G_F} < {\varepsilon _0} ,{G_B} > {\varepsilon _0} } \right\} \\
& = \int_{{\varepsilon _1}}^{{\varepsilon _0}} {{f_{{G_B}}}\left( y \right)\left( {{F_{{G_F}}}\left( y \right) - F_{{G_F}}^\infty \left( {\frac{{{\Theta _B}}}{\rho }} \right)} \right)dy}  \\
&+ \left( {1 - F_{{G_F}}^\infty \left( {{\varepsilon _0}} \right)} \right)\left( {F_{{G_F}}^\infty \left( {{\varepsilon _0}} \right) - F_{{G_F}}^\infty \left( {\frac{{{\Theta _B}}}{\rho }} \right)} \right)\\
& \approx F_{{G_B}}^\infty \left( {{\varepsilon _0}} \right) - F_{{G_B}}^\infty \left( {{\varepsilon _1}} \right) + F_{{G_F}}^\infty \left( {{\varepsilon _0}} \right) \\
&- F_{{G_F}}^\infty \left( {\frac{{{\Theta _B}}}{\rho }} \right) - {A_1}\sum\limits_{i = 0}^{m - 1} {\frac{{{{\left( {{\lambda _F}} \right)}^i}\Phi _2^\infty }}{{i!}}}\\
& = \frac{{\lambda _B^m\left( {\varepsilon _0^m - \varepsilon _1^m} \right)}}{{m!}} + \frac{{\lambda _F^m}}{{m!}}\left( {\varepsilon _0^m - \frac{{\Theta _B^m}}{{{\rho ^m}}}} \right) \\
&- {A_1}\sum\limits_{i = 0}^{m - 1} {\frac{{\lambda _F^i\Phi _2^\infty }}{{i!}}}
\end{aligned}
\end{equation}
where $\Phi _2^\infty = \frac{{\varepsilon _0^{i + m} - \varepsilon _1^{i + m}}}{{i + m}}$.

\end{appendices}

\end{document}